\documentclass[11pt,a4paper]{article}
\usepackage{authblk}
\usepackage{latexsym}
\usepackage{subfigure}
\usepackage{amsmath}
\usepackage{enumerate}
\usepackage{amssymb}
\usepackage[mathscr]{eucal}
\usepackage{amsthm}
\usepackage{color}
\usepackage[small]{caption}
\usepackage[hypertexnames=false,colorlinks=true,linkcolor=blue,citecolor=blue,
            bookmarksopen=false,bookmarks=false,
            pdfstartview=XYZ,pdffitwindow=true,pdfcenterwindow=true]{hyperref}
\usepackage[a4paper,text={6.0in,9.0in},centering,
            includefoot,foot=0.6in]{geometry}
\allowdisplaybreaks[1]
\usepackage{ifpdf}

\ifpdf
\usepackage[pdftex]{graphicx}
\else
\usepackage{graphicx}
\fi
\title{A numerical closure approach for kinetic models of polymeric fluids: exploring closure relations for FENE dumbbells}
\author{Giovanni Samaey\thanks{Department of Computer Science, K.U. Leuven, Celestijnenlaan 200A, B-3001 Leuven, Belgium {\tt giovanni.samaey@cs.kuleuven.be}}, Tony Leli\`evre\thanks{CERMICS, Ecole des Ponts ParisTech, 6 et 8 avenue Blaise Pascal, Cité Descartes - Champs sur Marne, 77455 Marne la Vallée Cedex 2, France {\tt lelievre@cermics.enpc.fr}} and Vincent Legat\thanks{Department of Mechanical Engineering and iMMC, U.C. Louvain, Avenue Georges Lemaître, 4, B-1348 Louvain-la-Neuve, Belgium \tt{vincent.legat@uclouvain.be}}}

\date{\today}

\DeclareMathOperator*{\argmin}{{\rm argmin}}

\begin{document}

\ifpdf
\DeclareGraphicsExtensions{.pdf, .jpg, .tif}
\else
\DeclareGraphicsExtensions{.eps, .jpg}
\fi

\newtheorem{thm}{Theorem}[section]
\newtheorem{cor}[thm]{Corollary}
\newtheorem{lem}[thm]{Lemma}
\newtheorem{pro}[thm]{Proposition}
\newtheorem{alg}[thm]{Algorithm}
\newtheorem{cond}[thm]{Condition}
\theoremstyle{definition}
\newtheorem{model}[thm]{Model problem}
\newtheorem{defn}[thm]{Definition}
\newtheorem{rem}[thm]{Remark}
\newtheorem{ex}[thm]{Example}

\numberwithin{equation}{section}

\newcommand{\We}{\textrm{We}}
\newcommand{\Id}{\textrm{Id}}
\renewcommand{\d}{\delta}
\newcommand{\De}{\textrm{De}}
\renewcommand{\Re}{\textrm{Re}}
\newcommand{\X}{\mathbf{X}}
\renewcommand{\u}{\mathbf{u}}
\newcommand{\F}{\mathbf{F}}
\newcommand{\x}{\mathbf{x}}
\newcommand{\W}{\mathbf{W}}
\newcommand{\M}{\mathbf{M}}
\newcommand{\R}{\mathcal{R}}
\newcommand{\system}[1]{\left\{ \begin{alignedat}{2}#1\end{alignedat}\right.}
\newcommand{\comment}[1]{ {\large *** {\bf #1}***}}

\newcommand{\RR}{\mathbb{R}}

\maketitle

\begin{abstract}
We propose a numerical procedure to study closure approximations for FENE dumbbells in terms of chosen macroscopic state variables, enabling to test straightforwardly which macroscopic state variables should be included to build good closures. The method involves the reconstruction of a polymer distribution related to the conditional equilibrium of a microscopic Monte Carlo simulation, conditioned upon the desired macroscopic state.  We describe the procedure in detail, give numerical results for several strategies to define the set of macroscopic state variables, and show that the resulting closures are related to those obtained by a so-called quasi-equilibrium approximation~\cite{Ilg:2002p10825}.
\end{abstract}

\section{Introduction}

The simulation of dilute solutions of polymers in a Newtonian solvent is a challenging modelling and numerical problem, since deformation of the polymer molecules causes stresses that result in macroscopic non-Newtonian rheological behavior. One approach is to couple the macroscopic fluid flow equations to a microscopic model for the polymers, a so-called micro-macro model~\cite{Hulsen:1997p7027,Laso:1993p10000,le-bris-lelievre-09}.
The simplest microscopic models, that we will use in this paper, describe the individual polymers as non-interacting dumbbells, consisting of two beads connected by a spring that models intramolecular interaction. The state of the polymer chain is described by the end-to-end vector $\X_t$ that connects both beads whose evolution is modelled using a stochastic differential equation (SDE):
\begin{equation}\label{eq:intro-micro}
	d\X_t+\u \cdot \nabla_x \X_t \, dt=\left[\nabla_x \u \X_t-\dfrac{2}{\zeta}\F(\X_t)\right]dt + \sqrt{\dfrac{4 k_B T}{\zeta}}d\W_t,
\end{equation}
where $\u$ is the velocity field of the solvent, $\zeta$ is a friction coefficient, $T$ is the temperature, $k_B$ is the Boltzmann constant, and $\W_t$ is a standard multidimensional Brownian motion. This model takes into account Stokes drag (due to the solvent velocity field), a spring force $\F$ and Brownian motion (due to collisions with solvent molecules). The left-hand side of Equation \eqref{eq:intro-micro} is the convective derivative. Note that the stochastic process $\X_t$ implicitly depends on the space variable $x$.

To specify the microscopic model \eqref{eq:intro-micro} completely, we need to define the spring force. This force can be more or less complicated, depending on the effects taken into account. The simplest model is the Hookean dumbbell model for which the spring is linear elastic:
\[
\F(\X)=H\X,
\]
with $H$ a spring constant. Another model, which is the focus of this paper and which is known to yield better agreement with experiments, is the finitely extensible nonlinear elastic (FENE) force~\cite{Bird:1987p11088}:
\begin{equation}\label{e:fene}
\F(\X)=\dfrac{H\X}{1-\|\X\|^2/(bk_BT/H)},
\end{equation}
where $b$ is a nondimensional parameter related to the maximal polymer length.

In the macroscopic part of the model, the evolution of the solvent velocity and pressure fields $\u$ and $p$ is modeled by mass and momentum conservation equations:
\begin{equation}\label{eq:intro-macro}
\left\{
\begin{aligned}
&	\rho \left(\dfrac{\partial \u}{\partial t}+\u\cdot\nabla_x  \u\right)=\eta_s\Delta_x \u - \nabla_x p + \textrm{div}_x(\tau_p),\\
&	\textrm{div}_x(\u)=0,
\end{aligned}
\right.
\end{equation}
with $\rho$ the density and $\eta_s$ the viscosity.
Equation \eqref{eq:intro-macro} contains an additional stress tensor $\tau_p$ due to polymer deformation, which is given via the classical Kramers' expression
\begin{equation}\label{eq:intro-stress}
\tau_p(x,t) = n\langle \X_t\otimes \F(\X_t)\rangle-nk_BT \, \Id.
\end{equation}
Here, $n$ is the polymer concentration and $\langle \cdot \rangle$ denotes the expectation over configuration space, which is approximated in practice by an empirical mean over a very large ensemble of realizations of $\X_t$, solutions to \eqref{eq:intro-micro}.

One thus obtains a coupled system \eqref{eq:intro-micro}--\eqref{eq:intro-macro}--\eqref{eq:intro-stress} that we rewrite in a non-dimensional form as (see for example \cite{Jourdain:2004p7045}):
\begin{align}
&	\Re \left(\dfrac{\partial \u}{\partial t}+\u\cdot\nabla_x  \u\right)=(1-\epsilon)\Delta_x \u - \nabla_x p + \textrm{div}_x(\tau_p), \label{eq:NS1}\\
&	\textrm{div}(\u)=0, \label{eq:NS2}\\
&\tau_p = \dfrac{\epsilon}{\We} \Big( \left\langle \X_t\otimes \F(\X_t)\right\rangle-\Id \Big), \label{eq:mic1} \\
&	d\X_t+\u \cdot \nabla_x \X_t \, dt=\left[\nabla_x \u \, \X_t-\dfrac{1}{2\We}\F(\X_t)\right]dt + \dfrac{1}{\sqrt{\We}}d\W_t, \label{eq:mic2}
\end{align}
where the nondimensional parameters are:
\begin{equation}
	\Re = \dfrac{\rho U L}{\eta},\qquad \We = \dfrac{\lambda U}{L}, \qquad \epsilon =\dfrac{\eta_p}{\eta}.
\end{equation}
Here, $U$ is a characteristic velocity, $L=\sqrt{k_B T/H}$ denotes a characteristic length,  $\lambda=\zeta/4H$ is a characteristic relaxation time for the polymers and $\eta_p=nk_BT\lambda$ is a viscosity associated to the polymers. The total viscosity is $\eta = \eta_p+\eta_s$.  
The parameters $\Re$ and $\We$ are the Reynolds and Weissenberg number, respectively. The nondimensional Hookean and FENE forces write respectively:
\begin{equation}\label{eq:springs}
	\F_{HOOK}(\X)=\X, \qquad \F_{FENE}(\X)=\dfrac{\X}{1-\|\X\|^2/b}.
\end{equation}

The microscopic part of the model, {\em i.e.}~\eqref{eq:mic1}--\eqref{eq:mic2}, can equivalently be described by a diffusion equation that governs the evolution of the probability distribution $\varphi(\X,x,t)$ of the random variable $\X_t$ (considered at point $x$ in physical space):
\begin{equation}\label{eq:intro-fp}
	\frac{\partial\varphi}{\partial t}+\u \cdot \nabla_x \varphi = \frac{1}{2 \We} \Delta_\X \varphi - {\rm div}_\X \left(\nabla_x \u \, \X \, \varphi\right)+\frac{1}{2 \We}{\rm div}_\X \left(\F(\X)\varphi\right),
\end{equation} 
The expectation in~\eqref{eq:mic1} then becomes an average with respect to the probability measure $\varphi(\X,x,t) \, d \X$: \begin{equation}\label{eq:mic1fp}
\tau_p(x,t)=\frac{\epsilon}{\We} \left( \int \X\otimes F(\X) \, \varphi(\X,x,t)d\X -  \Id \right). 
\end{equation}
We refer for example to \cite{Bird:1987p11088,Doi:1988p9812,Ott96} for more details on the physical background and more complicated models. 

A numerical simulation of the coupled system \eqref{eq:NS1}--\eqref{eq:mic2} is very expensive, since one needs to obtain the non-Newtonian stress tensor $\tau_p$ at each space-time discretization node.  Several approaches have been proposed in the literature \cite{Keunings:2004p49,le-bris-lelievre-09}.
A first approach is a deterministic micro-macro simulation. Here, one couples the Fokker--Planck equation \eqref{eq:intro-fp}--\eqref{eq:mic1fp} with the Navier--Stokes equations~\eqref{eq:NS1}--\eqref{eq:NS2}. The main drawback of these methods is their high computational cost, due to the high-dimensionality of the function $\varphi$ (which depends on seven scalar variables $(\X,x,t)$ in dimension 3). This difficulty becomes all the more severe when more refined models involving higher dimensional microscopic variables $\X_t$ are used to describe the polymers. Specialized techniques are currently being developed; see~{\em e.g.}~\cite{Ammar:2006p10934,Ammar:2007p10932,delaunay-lozinski-owens-07}.
The micro-macro simulation can also be performed stochastically. One then  discretizes the macroscopic fields (velocity, pressure, stress) on a mesh, and supplements the (macroscopic) discretization of the Navier-Stokes equations with a stochastic simulation of an ensemble of polymers using a discretization of the SDE~\eqref{eq:mic2}, see~\cite{Hulsen:1997p7027,Laso:1993p10000}. Methods have been proposed to obtain sufficiently low-variance results \cite{Bonvin:1999p4981,Hulsen:1997p7027,Jourdain:2004p7045}.

Due to the very high computational cost of micro-macro simulations, another route which has been followed (see~{\em e.g.}~\cite{Herrchen:1997p8915,Hyon:2008p9897,Keunings:1997p9982,Lielens:1998p6790,Lielens:1999p9945,Sizaire:1999p9912}) is to look for an approximate closure at the macroscopic level, namely a model of the form:
\begin{align}
	&\dfrac{\partial \M}{\partial t}+\u\cdot\nabla_x \M=\mathcal{H}(\M,\nabla_x \u),\label{eq:closure_macro}
	\\
&\tau_p=T(\M), \label{eq:constitutive}
\end{align}
which is close to the microscopic model~\eqref{eq:mic1}--\eqref{eq:mic2}. Here $\M$ denotes an ensemble of macroscopic state variables that depend on time and space. A basic example of such a macroscopic model is the Oldroyd-B model~\cite{Bird:1987p11088}, which is actually equivalent to the microscopic model~\eqref{eq:mic1}--\eqref{eq:mic2} for a linear force $\F(\X)=\F_{HOOK}= \X$.
In this case, one can obtain a closed equation on the so-called conformation tensor ${\boldsymbol \sigma}(t)=(\sigma_{i,j}(t))_{i,j=1}^d$, with $d$ the number of space dimensions, and $ \sigma_{i,j}(t)= \langle (X_i)_t (X_j)_t\rangle$, in which $(X_{i})_t$, resp.~$(X_{j})_t$, represent the corresponding component of $\X_t$. This yields the equation :
$$
\partial_t \tau_p + \u \cdot \nabla_x \tau_p = \nabla_x \u \, \tau_p + \tau_p \,\nabla_x \u^T + \frac{\epsilon}{\We} ( \nabla_x \u + \nabla_x \u^T) - \frac{1}{\We} \tau_p.
$$On the other hand, for the FENE model, no equivalent closed macroscopic model is known, and one has to resort to approximate closures to obtain macroscopic equations (see Section~\ref{sec:closure}). The basic idea is to approximate the polymer distribution by a so-called \emph{canonical distribution function},
which is determined using only the macroscopic state variables $\M$ (typically low-order moments of the distribution). The microscopic evolution law~\eqref{eq:mic2} (or \eqref{eq:intro-fp}) is then replaced by a set of equations \eqref{eq:closure_macro} for the evolution of the macroscopic state variables $\M$, combined with a constitutive equation \eqref{eq:constitutive} for the stress.
% In this paper, we will in particular consider the method proposed in~\cite{Ilg:2002p10825}, where the canonical distribution is obtained through a {\em quasi-equilibrium condition}, namely by maximizing an entropy function under constraints prescribed by the macroscopic state variables $\M$ (see Section~\ref{sec:ott}). 
While such approximate macroscopic models are desirable, at least from a computational point of view, it is however not always clear how to quantify the effects of the introduced approximations on the accuracy of the simulation, and how to choose the macroscopic state variables $\M$. 
%One contribution of this paper is to propose a numerical procedure to test rapidly which macroscopic state variables $\M$ should be taken into account to get a good closure approximation.

Recently, there has been quite some interest in the development of computational methods that aim at accelerating micro-macro simulation using on-the-fly numerical closure approximations. We mention equation-free \cite{KevrGearHymKevrRunTheo03,Kevrekidis:2009p7484} and heterogeneous multiscale methods (HMM) \cite{EEng03,E:2007p3747}. 
% In the \emph{equation-free framework} \cite{KevrGearHymKevrRunTheo03,Kevrekidis:2009p7484}, one assumes the \emph{existence} of a closed macroscopic (macroscopic) model in terms of a few macroscopic state variables, here the low-order moments $\M$. However, instead of deriving an approximate macroscopic model, one constructs a computational superstructure, wrapped around a microscopic (microscopic) simulation, here using the stochastic model \eqref{eq:intro-micro}. 
% A key tool is the \emph{coarse time-stepper}, which implements a time step of an \emph{unavailable} (in closed form)
% macroscopic model \eqref{eq:closure_macro} as a three-step procedure: 
% (i)	\emph{lifting}, {\em {\em i.e.}}\ 
% creation of initial conditions for the microscopic model, conditioned
% upon the macroscopic state at $t^*$; 
% (ii) \emph{simulation}, using the microscopic
% model, over a time interval $[t^*,t^*+\delta t]$; and 
% (iii) \emph{restriction}, {\em {\em i.e.}}\  observation (estimation) of the macroscopic state
% at $t^*+\delta t$. In the strongly related \emph{heterogeneous multiscale methods}, the unavailable macroscopic model is assumed to be of some \emph{general form}, containing unknown quantities that are then estimated from microscopic simulation \cite{EEng03,HMMReview}. 
In both approaches, a crucial step is to define an operator that generates a microscopic state corresponding to a given macroscopic state; this is actually equivalent to prescribing the closure approximation. This step is called \emph{lifting} in the equation-free framework, and \emph{reconstruction} in HMM. Inspired by these methods, the present paper studies in detail the question of lifting/reconstruction for the particular problem of micro-macro models for polymeric fluids; the procedure we propose, however, could be applied to many multiscale models. Specifically, we propose a computational procedure to reconstruct an ensemble of $N$ polymers consistently with a given macroscopic state $\M$, and we examine the errors that are introduced in the macroscopic evolution by {\em numerically enforcing closure} upon the selected macroscopic state variables. For convenience of exposition and illustration, we restrict ourselves to one-dimensional simulations with pre-imposed (time-dependent) velocity fields,  {\em i.e.}~equations \eqref{eq:mic1}--\eqref{eq:mic2} with given $\u(x,t)$, at one specific point $x$ in space. However, we emphasize that the numerical method can be used likewise for 2D or 3D situations, as well as for the closure approximation for the coupled problem~\eqref{eq:NS1}--\eqref{eq:mic2}.

The main contributions of the present paper are twofold:
\begin{itemize}
\item From a modelling viewpoint, we propose a numerical closure strategy that enables to easily explore which sets of macroscopic state variables should be chosen to get good closure approximations. Various strategies are proposed and tested.
\item From a theoretical viewpoint, we show the relation between this numerical closure strategy and the so-called quasi-equilibrium method proposed in~\cite{Ilg:2002p10825}, which relies on an entropy minimization principle.
\end{itemize}

The paper is organized as follows. In Section~\ref{sec:model}, we give some more detail on the FENE model and the existing literature on closure approximations. In Section~\ref{sec:ef}, we propose a numerical closure approximation based on constrained SDE simulations~\cite{Lelievre:2008p9419}, which is very flexible, and enables to explore the error introduced by the closure for various sets of macroscopic state variables $\M$. This numerical closure approximation is shown to be optimal in the sense that, when applied to a microscopic model which has an equivalent macroscopic model, it indeed yields the macroscopic model (Section~\ref{sec:fene-p}). Moreover, we show in Section~\ref{sec:eq_ott} that, in some specific cases, it is closely related to the closure approximation based on a quasi-equilibrium condition introduced in~\cite{Ilg:2002p10825}. Finally, we test the numerical closure using a number of different strategies to define the macroscopic state variables $\M$ (Section~\ref{sec:num-fene}). We first perform numerical experiments to assess the capability of the selected macroscopic state variables to \emph{recover the desired polymer distributions} in strong flow regimes. Second, we study if the procedure is able to \emph{correctly capture macroscopic evolution}. While accelerating microscopic simulation is not the primary purpose of the present paper, we give some remarks in this respect in Section~\ref{sec:concl}, where we briefly discuss the main results and give some directions for future research. 

\section{The FENE model and closure approximations}\label{sec:model}

\subsection{FENE dumbbells: discretization and a one-dimensional version}

As mentioned above, we consider polymer simulations with FENE dumbbells subject to a pre-imposed (time-dependent) velocity field. Thus, in the remainder of the paper, unless explicitly stated otherwise, the force is the FENE force, see~\eqref{eq:springs}~:
$$\F=\F_{FENE}.$$
Using the characteristic method to integrate the convective derivative in~\eqref{eq:mic2} (Lagrangian frame), the equations of interest reduce to:
\begin{equation}\label{eq:mic2_lag}
\begin{aligned}  
&\tau_p = \dfrac{\epsilon}{\We} \Big( \left\langle \X_t\otimes \F(\X_t)\right\rangle-\Id \Big),  \\
&	d\X_t =\left[ {\boldsymbol \kappa}(t) \X_t-\dfrac{1}{2\We}\F(\X_t)\right]dt + \dfrac{1}{\sqrt{\We}}d\W_t,
\end{aligned}
\end{equation}
where $\X_t$ now depends on the foot of the characteristic rather than on the Eulerian space position $x$, and ${\boldsymbol \kappa}$ is the velocity gradient (along the trajectory). Unless stated otherwise, we will work with a one-dimensional version of this equation,
\begin{equation}\label{eq:1D-preimposed}
\left\{
\begin{aligned}
& \tau_p= \dfrac{\epsilon}{\We} \Big( \left\langle X_t \, F(X_t)\right\rangle-1 \Big),\\
&dX_t=\left[\kappa(t)X_t - \dfrac{1}{2\We}F(X_t)\right]dt + \dfrac{1}{\sqrt{\We}}dW_t,
\end{aligned}
\right.
\end{equation}
keeping in mind that the algorithm, as well as its analysis and implementation extend straightforwardly to higher dimensions.
Note that $\kappa(t)$ is here a given one-dimensional time-dependent function and $F$ denotes a one-dimensional version of the FENE force, see~\eqref{eq:springs}, namely $$F(X)=\dfrac{X}{1-X^2/b}.$$
Such a one-dimensional framework has also been used in~\cite{Keunings:1997p9982} for example, to assess the influence of the Peterlin approximation (see Section~\ref{sec:closure}) on transient behaviour.

Concerning discretization methods, we use a classical Euler-Maruyama scheme~\cite{KloPla92} with a Monte Carlo method:
\begin{equation}\label{eq:model-discr}
\left\{
\begin{aligned}
&\tau_{p}^k=\dfrac{\epsilon}{\We} \left(  \dfrac{1}{N}\sum_{n=1}^N \left(X^{n,k} \, F(X^{n,k})\right) -1 \right) ,\\
&	X^{n,k+1}=X^{n,k}+\left[ \kappa(t^k) X^{n,k} - \dfrac{1}{2\We}F(X^{n,k})\right]\d t + \dfrac{1}{\sqrt{\We}}\sqrt{\d t}\; \xi^{n,k},
\end{aligned}
\right.
\end{equation}
where  the indices $n$ and $k$ denote respectively realization index and time index, $t^k= k \d t$ and $\xi^{n,k}$ are i.i.d. normal random variables.

For convenience, we introduce a short-hand notation for the discretization scheme of the SDE in~\eqref{eq:model-discr},
\begin{equation}
	\mathcal{X}^{k+1}=s_{\X}(\mathcal{X}^k,\kappa(t^k),\d t),
\end{equation}
where $\mathcal{X}=\{X^n\}_{n=1}^N$ is the ensemble of $N$ realizations, and $\kappa(t^k)$ indicates explicitly the value of the velocity gradient in~\eqref{eq:model-discr} that is considered over the time interval of size~$\d t$.

Theoretically, it can be shown that (for sufficiently large $b$), the norm of the end-to-end vector in~\eqref{eq:mic2} or~\eqref{eq:1D-preimposed} (recall that $F=F_{FENE}$) cannot exceed the maximal value $\sqrt{b}$~\cite{jourdain-lelievre-02}. However, the discretization scheme~\eqref{eq:model-discr} might yield spring lengths beyond this maximal value. There are two ways to deal with this problem~\cite[Section 4.3.2]{Ott96}.  The first is via an accept-reject method, in which, for each polymer, the state after each time step is rejected if $|\X|^2>(1-\sqrt{\d t})b$, and a new random number is tried until acceptance.  Alternatively, one could use a semi-implicit predictor-corrector method. In this text, we choose the accept-reject strategy.

\subsection{Closure approximations for FENE dumbbells}\label{sec:closure}

We now briefly discuss the derivation of closure approximations of the type~\eqref{eq:closure_macro}-\eqref{eq:constitutive} for the FENE model.

One closure approximation is the Peterlin pre-averaging \cite{BIRD:1980p11107}.  Here, one constructs an approximation for the FENE model by defining the spring force as (compare with~\eqref{eq:springs})
\begin{equation}\label{e:fenep}
F_{FENE-P}(X)=\dfrac{X}{1-\langle X^2\rangle/b}.
\end{equation}
As a consequence, only the mean square length of the ensemble of polymers is constrained to remain smaller than $\sqrt{b}$, whereas the length of individual polymers may exceed this value. The interest of FENE-P dumbbells is that, as for Hookean dumbbells, a closed equation can be derived on the conformation tensor $\sigma = \langle X^2_t\rangle$, and thus a macroscopic model is obtained:
\begin{equation}\label{eq:FENE-P}
\begin{aligned}
&\partial_t \sigma + u  \nabla_x \sigma = 
2 \sigma  \nabla_x u - \frac{1}{\We} \frac{ \sigma }{1 -
    {\rm tr}(\sigma)/b} + \frac{1}{\We} .\\
&\tau_p=\frac{\epsilon}{\We} \left( \frac{ \sigma }{1 -
    {\rm tr}(\sigma)/b} - 1 \right),
\end{aligned}
\end{equation}
It has been shown in~\cite{Herrchen:1997p8915,Keunings:1997p9982} that the Peterlin approximation has a profound impact on transient behaviour in complex flows, compared to the original FENE model.

Let us now discuss more generally closure approximations of the type~\eqref{eq:closure_macro}. For the sake of clarity, and without loss of generality, we restrict ourselves to the one-dimensional case~\eqref{eq:1D-preimposed}.

Consider starting from a number $L$ of macroscopic state variables, $\M=\left\{ M_l\right\}_{l=1}^L$, which are defined as configuration space averages of functions $m_l$ of the configuration $X_t$,
\begin{equation}\label{eq:state-vars}
M_l(t) = \left\langle m_l(X_t)\right\rangle.
\end{equation} 
The goal is to obtain a closed system of $L$ evolution equations~\eqref{eq:closure_macro} for the state variables~$\M$, complemented with a constitutive equation \eqref{eq:constitutive} for $\tau_p$ as a function of these macroscopic state variables. 	

Using It\^o calculus, one can easily obtain the following equation of state for the macroscopic state variables,
\begin{equation}\label{eq:eq-of-state}
	\frac{d M_l}{d t} = \kappa(t) \underbrace{\left\langle X_t \, \frac{d m_l }{d X}(X_t)\right\rangle}_{M_l^D}-\frac{1}{2\We }\underbrace{\left\langle F(X_t) \,\frac{d m_l}{d X}(X_t)\right\rangle}_{M_l^C}+\frac{1}{2 \We}\underbrace{\left\langle\frac{d^2 m_l}{d X^2}(X_t)\right\rangle}_{M_l^B},
\end{equation}
in which the macroscopic state variables $M_l^{\{D,C,B\}}$ account for hydrodynamic drag, connector force and Brownian motion, respectively. Of course, in general, many of these macroscopic state variables $M_l^{\{D,C,B\}}$ are not functions of the initially chosen macroscopic state variables $\left\{ M_l\right\}_{l=1}^L$. One can  write evolution equations for these new state variables, which in turn will create additional state variables but this procedure typically goes on endlessly. At some point, one has to stop, and try to approximate the state variables for which no evolution equation is available by writing them as a function of other (already available) state variables. By adding such \emph{closure relations}, one obtains an explicit, but approximate, closed system of evolution equations.

Any closed macroscopic model needs to (i)  define the set of macroscopic state variables $\M=\left\{ M_l\right\}_{l=1}^L$, and (ii) provide a way of evaluating the remaining state variables $M_l^{\{D,C,B\}}$ in the evolution equation as a function of $\M$. In the literature, item (i) is generally addressed by considering a hierarchy of even moments, {\em i.e.}\ $M_l=\langle X^{2l}\rangle$ where $l=1,\ldots,L$ (all the odd moments are zero for reasons of symmetry). Note that these become tensors in higher space dimensions. The corresponding evolution equations~\eqref{eq:eq-of-state} are then given as:
\begin{equation}\label{eq:eos-even}
		\frac{d M_l}{d t}= 2l \, \kappa(t) \, M_l  -\frac{1}{2 \We}M_l^C + \frac{l(2l-1)}{\We} M_{l-1}, 
\end{equation}
with $M_0=1$. In order to complete (ii), one needs to provide approximations for the  new additional macroscopic state variables $\left\{M_l^C\right\}_{l=1}^L$. Note that, in particular, one of this new additional macroscopic variable $M_1^C$ is also required to obtain the constitutive relation~\eqref{eq:constitutive} for $\tau_p$. One strategy to approximate $\left\{M_l^C\right\}_{l=1}^L$ is to propose a probability distribution $\varphi_{\M}(X)$ (called a \emph{canonical distribution function}) that is parameterized by the selected macroscopic state variables, and to compute $\left\{M_l^C\right\}_{l=1}^L$ in the evolution equations \eqref{eq:eq-of-state} as the expectation with respect to this canonical distribution function. Note that $\varphi_{\M}(X)$ depends on time only through the dependency of $\M$ on the time variable. The rationale behind this approach is that the better one can approximate the microscopic distribution function, the more reliable the obtained macroscopic model should be. 

In \cite{Lielens:1998p6790,Lielens:1999p9945}, approximate closures for $M_l^C$ are obtained by restricting the space of admissible distribution functions to linear combinations of $L$ \emph{canonical basis functions}. 
Based on this approach, several closures have been proposed; see \cite{Lielens:1998p6790} for more details on the one-dimensional setting \eqref{eq:1D-preimposed} and \cite{Lielens:1999p9945} for the general three-dimensional case. A related approach is described in \cite{Du:2006p11012,Hyon:2008p9897,Yu:2005p11011}. Another route is described in the following section.

\subsection{Quasi-equilibrium approximations}\label{sec:ott}

A particularly interesting approach is proposed in~\cite{Ilg:2002p10825}. It consists in defining a so-called quasi-equilibrium canonical distribution function $\varphi^{QE}_\M$ via a constrained entropy optimization problem:
\begin{equation}\label{eq:phiM_ott}
\varphi^{QE}_\M=\argmin_{\varphi \in \Omega_\M} \int \varphi \ln \left( \frac{\varphi}{\varphi_{eq}} \right),
\end{equation}
where $\Omega_\M$ is defined as the set of all probability density functions, for which the average of $m_l$ is indeed $M_l$:
\begin{equation}\label{eq:omega_M}
\Omega_\M= \left\{ \varphi(X), \, \varphi \ge 0 , \, \int \varphi(X) \, dX =1, \, \int m_l(X) \varphi(X) \, dX = M_l , \, l = 1, \ldots,L \right\}.
\end{equation}
In~\eqref{eq:phiM_ott}, $\varphi_{eq}$ is defined as the equilibrium distribution for the polymer configuration, for zero velocity field. In particular, for FENE dumbbells, it writes:
$$\varphi_{eq}(X)= Z^{-1} \left(1- X^2/b\right)^{b/2},$$
where $\displaystyle{Z=\int_{|X|\le \sqrt{b}} \left(1- X^2/b\right)^{b/2} dX}$.

The rationale behind this approximation is to assume a separation of time scales between the (supposedly fast) relaxation towards the quasi-equilibrium distribution and the (supposedly much slower) evolution of the macroscopic state variables.
 
An explicit expression of the solution to~\eqref{eq:phiM_ott} can be obtained as:
\begin{equation}\label{eq:QE}
\varphi^{QE}_\M(X) = Z_\M^{-1} \, \varphi_{eq}(X) \exp \left( \sum_{l=1}^L \lambda_l \, m_l(X) \right),
\end{equation}
where $\displaystyle{Z^{QE}_\M= \int \varphi_{eq}(X) \exp \left( \sum_{l=1}^L \lambda_l \, m_l(X) \right) \, dX}$ and the set of Lagrange multipliers $\Lambda=\{\lambda_l\}_{l=1}^L$ are determined by the constraints $\displaystyle \int m_l(X) \varphi^{QE}_\M(X) \, dX = M_l$. 

While the Lagrange multipliers depend only on the macroscopic state $\M$, the relation $\Lambda(\M)$ can often not be obtained analytically. Therefore, in \cite{Ilg:2002p10825}, a numerical procedure is proposed to simulate the resulting closed macroscopic model. We will show below (see Section~\ref{sec:eq_ott}) that the numerical closure approximation technique that we propose (see Section~\ref{sec:ef}) is closely related to this method, and that it may be considered (for a slightly modified version) as a different numerical strategy to obtain quasi-equilibrium closure approximations.

\section{Numerical method}\label{sec:ef}

In this section, we propose to mimic the evolution of the corresponding unavailable macroscopic model via a \emph{coarse time-stepper} \cite{KevrGearHymKevrRunTheo03,Kevrekidis:2009p7484}.

\subsection{The lifting and restriction operators}

Consider the evolution of an ensemble of polymers in a pre-imposed velocity field and define a set of macroscopic state variables $\M$ which are believed to represent the underlying (microscopic) polymer distribution sufficiently accurately.  
We introduce two operators that make the transition between
microscopic and macroscopic state variables.
We define a \emph{lifting operator},
\begin{equation}\label{eq:intro_lifting}
\mathcal{L}: \M \mapsto \mathcal{X},
\end{equation}
which maps a macroscopic state to an ensemble of $N$ polymer configurations, and the associated
\emph{restriction operator},
\begin{equation}\label{eq:intro_restriction}
\mathcal{R}: \mathcal{X} \mapsto \M,
\end{equation}
which maps an ensemble of configurations to the corresponding macroscopic state.
Note that we directly define the method at the discrete level over an ensemble of $N$ configurations (after Monte Carlo discretization). For a discussion in the limit of an infinitely large number of polymer configurations, we refer to Section~\ref{sec:lift_rest_limit}.

The restriction operator is readily defined using an empirical mean:
\begin{equation}\label{eq:R}
\mathcal{R}(\mathcal{X})=\{M_l=\mathcal{R}_l(\mathcal{X}) \}_{l=1}^L \text{ with } \mathcal{R}_l(\mathcal{X}) = \dfrac{1}{N}\sum_{n=1}^N m_l(X^n) \text{ for $l=1,\ldots,L$},
\end{equation}
where, we recall, $\mathcal{X}=\{X^n\}_{n=1}^N$ denotes the ensemble of configurations.

In the lifting step, we need to sample a reconstructed polymer distribution function, consistently with the given macroscopic state $\M(t^*)$ obtained at time $t^*$. To this end, we perform {\em a constrained simulation of an ensemble of polymers until equilibrium}, subject to the constraint that the macroscopic state remains constant and equal to $\M(t^*)$. More precisely, the constrained algorithm writes~\cite{Lelievre:2008p9419}:
\begin{equation} \label{eq:constrained}
\left\{
\begin{aligned}
&\mathcal{X}^{m+1}=s_\X(\mathcal{X}^m,\kappa(t^*),\d t)+\sum_{l=1}^L \lambda_l \nabla_{\mathcal{X}} \mathcal{R}_l (\mathcal{X}^m),\\
&\text{ with $\Lambda \in {\mathbb R}^L$ such that $\mathcal{R}_l (\mathcal{X}^{m+1}) = M_l(t^*)$ for $l=1,\ldots,L$.}
\end{aligned}
\right.
\end{equation}
It thus consists successively in an unconstrained Euler-Maruyama step, followed by a projection step to satisfy the constraint. In each constrained time step, the projection is done by solving the nonlinear system 
\begin{equation}\label{eq:lm-newton}
\mathcal{R}_l\left(\mathcal{X}^{m+1}(\Lambda;\mathcal{X}^{m},\delta t)\right)=M_l(t^*), \qquad \text{ for } l=1,\ldots,L,
\end{equation}
for the unknown Lagrange multipliers $\Lambda$ using Newton's method.  In \eqref{eq:lm-newton}, we have made explicit that the state $\mathcal{X}^{m+1}$ depends on the unknown Lagrange multipliers, as well as on (known) $\mathcal{X}^m$ and $\delta t$.
During the constrained simulation, an accept-reject strategy is applied on the combined evolution and projection operation, {\em i.e.}\ if, 
during projection, the state of a polymer would become unphysical, we reject the trial move in the unconstrained Euler-Maruyama step and repeat the time step for this polymer, after which the projection of the ensemble is tried again.

The lifting operator is then defined as the ensemble $\mathcal{X}^{m_\infty}$ for a sufficiently large time index $m_\infty$, which is chosen such that~\eqref{eq:constrained} has reached an equilibrium distribution, 
\begin{equation}\label{eq:L-discr}
{\mathcal L}(\M) = \mathcal{X}^{m_\infty}.
\end{equation}
We will detail further on how $m_\infty$ is determined numerically when describing the computational experiments. 
For a precise definition of the lifting operator in terms of distributions (in the limit of an infinite number of configurations), we refer to Section~\ref{sec:lift_rest_limit}.

Of course, by construction one has the consistency property
$$\mathcal{R} \circ \mathcal{L}= \Id.$$

% Giovanni -> Tony: removed because we don't do projective integration (we never obtain new moments for which we do not have a distribution)
% The constrained simulation~\eqref{eq:constrained} requires an initial
% condition that satisfies the constraints. In general, during macroscopic
% projective integration, a distribution function is available for a nearby
% macroscopic state; this distribution can be projected onto the desired
% macroscopic state in the same way that the projection is done after. If no
% nearby distribution is available, one can start from a distribution that is
% obtained through an analytical closure approximation, such as the FENE-L model
% \cite{Lielens:1998p6790}. It is important to emphasize that, since we consider
% the stationary distribution of the constrained simulation, the initial
% condition does not affect the result (at least under the assumption of
% uniqueness of such a stationary distribution, which seems to be the case
% numerically).

\subsection{The numerical closure algorithm}\label{sec:algo}

Let us now make precise the complete algorithm. Given an initial condition for the macroscopic state variables $\M(t^*)$ at
time $t^*$, one time step of the coarse time-stepper consists of a three-step procedure:
\begin{itemize}
\item[(i)]	\emph{Lifting}, {\em i.e.}\ the
creation of initial conditions $$\mathcal{X}(t^*)=\mathcal{L}(\M(t^*))$$ for the microscopic model, consistently with the macroscopic state  $\M(t^*)$ at $t^*$.
\item[(ii)] \emph{Simulation} using the microscopic model over a time interval $[t^*,t^*+ K \delta t]$, where $K$ is the number of time steps, to get $\mathcal{X}(t^* + K \d t)$: for $k=0,\ldots,K-1$,
$$\mathcal{X}(t^*+ (k+1) \delta t)=s_\X \left(\mathcal{X}(t^*+k \delta t),\kappa(t^* + k \delta t), \delta t\right).$$
\item[(iii)] \emph{Restriction}, {\em i.e.}\ the observation (estimation) of the macroscopic state
at $t^*+K \delta t$:
$$\M(t^*+ K \delta t)=\mathcal{R}(\mathcal{X}(t^*+ K \delta t)).$$
\end{itemize}
In the following, we denote
$$\Delta t = K \delta t.$$

During the restriction step, the ensemble $\mathcal{X}(t^*+ \Delta t)$ is also used to get an estimate of the new value of the stress $$\tau_p(t^* + \Delta t)=  \frac{\epsilon}{\We} \left( \frac{1}{N} \sum_{n=1}^N X^n(t^*+ \Delta t) \, F(X^n(t^*+ \Delta t)) - 1 \right).$$

\subsection{The lifting and restriction operator in the continuous limit}\label{sec:lift_rest_limit}

The lifting and restriction operators which have been defined above depend on three discretization parameters: $N$ which is related to the Monte Carlo discretization (the operators have been defined for a finite ensemble of configurations), $\d t$ which is related to the time discretization in~\eqref{eq:constrained}, and $m_\infty$ which should be sufficiently large to reach a stationary state in~\eqref{eq:constrained}. In this section, we introduce the limiting operators $\overline{\mathcal L}$ and $\overline{\mathcal R}$ obtained in the limit $N \to \infty$, $\d t \to 0$ and $m_\infty \d t \to \infty$.

Note first that these operators are well-defined in terms of the probability distribution $\varphi$, rather than ensembles of configurations. More precisely,
the lifting operator $\overline{\mathcal L}$ consists in constructing a probability distribution $\varphi^{NC}_\M$ consistently with the macroscopic state variables $\M$ (using the notation of Sections~\ref{sec:closure}-\ref{sec:ott}),
\begin{equation}\label{eq:L-inf}
	\overline{\mathcal L}(\M)=\varphi_\M^{NC}(X),
\end{equation}
in which the superscript $NC$ stands for {\em numerical closure}.
Likewise, the restriction operator $\overline{\mathcal R}$ reduces a distribution to macroscopic state variables.

The restriction operator $\overline{\mathcal R}$ is simply an averaging operator, which computes the averages of $m_i$ with respect to the distribution $\varphi$ (compare with~\eqref{eq:R}):
\begin{equation}\label{eq:Rbar}
\overline{\mathcal{R}}(\varphi)=\{M_l=\overline{\mathcal{R}}_l(\varphi) \}_{l=1}^L \text{ with } \overline{\mathcal{R}}_l(\varphi) = \int m_l \varphi \text{ for $l=1,\ldots,L$},
\end{equation}
On the other hand, the lifting operator $\overline{\mathcal L}$ is more involved to define. When considering the continuous-in-time version of~\eqref{eq:constrained} in the limit of an infinite number of configurations, $N \to \infty$, it can be seen to be given by the one-dimensional marginal of the stationary state of the associated Fokker-Planck equation.

Let us make this statement precise.
For a fixed value $N$, the numerical scheme~\eqref{eq:model-discr} is a discretization of the following constrained Stratonovitch SDE on the ensemble ${\mathcal X}_t=\{X^n_t \}_{n=1}^N$ (see~\cite{Lelievre:2008p9419} and~\cite[Chapter 3]{lelievre-rousset-stoltz-10}):
\begin{equation}\label{eq:XNt}
d{\mathcal X}_t=P({\mathcal X}_t) \left[ \kappa(t^*) {\mathcal X}_t - \frac{1}{2 \We} F({\mathcal X}_t)  \right] \, dt + \frac{1}{\sqrt{\We}} P({\mathcal X}_t) \circ d \mathcal{W}_t,
\end{equation}
where, with a slight abuse of notation, $F({\mathcal X}_t)\equiv \left(F(X^n_t)\right)_{n=1}^N$, and $\mathcal{W}_t$ represents an $N$-dimensional Brownian motion. The projection operator $P({\mathcal X}_t)$ is defined by:
$$P({\mathcal X})= \Id - \sum_{i,j=1}^L G^{-1}_{i,j}({\mathcal X}) \nabla_{\mathcal X} {\mathcal R}_i({\mathcal X}) \otimes \nabla_{\mathcal X} {\mathcal R}_j({\mathcal X})$$
with $G^{-1}_{i,j}({\mathcal X})$ the inverse of the Gram matrix:
$$G_{i,j}({\mathcal X})=\nabla_{\mathcal X} {\mathcal R}_i({\mathcal X}) \cdot \nabla_{\mathcal X} {\mathcal R}_j({\mathcal X})$$
and $\circ$ denotes the Stratonovitch product. If we denote 
\begin{equation}\label{eq:sigma_M}
\Sigma(\M)=\{ {\mathcal X}, {\mathcal R}({\mathcal X}) = \M \}
\end{equation}
the submanifold of ${\mathcal X}$ at fixed values of the macroscopic state variables, then $P({\mathcal X})$ is the orthogonal projection operator onto the tangent space $T_{\mathcal X}\Sigma(\M)$ of $\Sigma(\M)$ at point $\mathcal X$. Thus, if ${\mathcal X}_0 \in \Sigma(\M)$, then, for all $t\ge 0$, ${\mathcal X}_t \in \Sigma(\M)$.

Let us denote $\psi^N(t,d{\mathcal X})$ the distribution of ${\mathcal X}_t$ satisfying~\eqref{eq:XNt}. Note that the components of ${\mathcal X}_t$ have all the same law, for symmetry reasons. Let us introduce the marginal of $\psi^N$ in the first variable:
\begin{equation}\label{eq:psi1N}
\psi^N_1(t,X^1)dX^1 = \int_{X^2,\ldots, X^N} \psi^N(t,dX^1,\ldots,dX^N).
\end{equation}
Then, $\varphi^{NC}_\M$ is defined as:
\begin{equation}\label{eq:psi_infty}
\varphi^{NC}_\M(X)=\lim_{N \to \infty} \lim_{t \to \infty} \psi^N_1(t,X).
\end{equation}
By a law of large numbers, it is expected that this distribution $\varphi^{NC}_\M$ is consistent with the fixed values of macroscopic state variables $\M$:
$$\varphi^{NC}_\M \in \Omega_{\M},$$
where $\Omega_{\M}$ is defined by~\eqref{eq:omega_M}.

We will discuss in Section~\ref{sec:eq_ott} how to get an analytical expression for $\varphi^{NC}_\M$, at least in some specific cases.

\subsection{Choice of the macroscopic state variables\label{sec:strat}}

For the FENE model, it appears that the first even moment $\langle X_t^2 \rangle$ is not sufficient to characterize the polymer distribution, and additional macroscopic state variables are needed. We will consider the macroscopic level to be determined by $L$ macroscopic state variables, $\M = \{M_l\}_{l=1}^{L}$, and we  consider the following strategies to select $M_l$, $l=1,\ldots, L$.

\paragraph{Strategy 1.} We consider a hierarchy of even moments of increasing order,
\begin{equation}
	M_l = \langle X_t^{2l}\rangle, \qquad l=1,\ldots, L.
\end{equation}
	
\paragraph{Strategy 2.} We consider a hierarchy of even moments of increasing order, and supplement the set of macroscopic state variables with the additional moments that appear in the corresponding evolution equations \eqref{eq:eos-even},
\begin{equation}
	\system{
	& 	M_l = \langle X_t^{2l}\rangle, &\\
	& M_{\tilde{L}/2+l}= M_l^C = 2l\langle F(X_t)\, X_t^{2l-1}\rangle, & \\
	}
\end{equation}
for $1\le l \le \tilde{L}/2$ where $\tilde{L}$ is assumed to be even. For FENE dumbbells, it can easily be checked that 
$$\tau_p=\dfrac{\epsilon}{\We} \left( M_1^C/2-1 \right),$$ and that all $M_l^C$, $l>1$ can be written as linear combinations of $M_l$, $l=1,\ldots, \tilde{L}/2$ and $\tau_p$.  Hence, this choice is equivalent to taking
\begin{equation}
	\system{
	& 	M_l = \langle X_t^{2l}\rangle, \qquad l=1,\ldots L-1&\\
	& M_{L}= \tau_p = \frac{\epsilon}{\We} \left( \langle X_t \, F(X_t)\rangle-1 \right) =\frac{\epsilon}{\We} \left( \left\langle\dfrac{X_t^2}{1- X_t^2/b}\right\rangle-1 \right), & \\
	}
\end{equation}
where $L=\tilde{L}/2+1$ denotes the number of linearly independent macroscopic state variables.

\paragraph{Strategy 3.} We again start from $M_1=\langle X_t^2 \rangle$. To add state variables, we write down the evolution equation for $M_1$, {\em i.e.}~\eqref{eq:eos-even} with $l=1$, and add all macroscopic state variables that appear in this equation. In this case, this amounts to adding the variable $M_2=M_1^C$.  We continue by writing down the evolution equation \eqref{eq:eq-of-state} for $M_2$, which, in turn, reveals additional state variables $M_2^{D,C,B}$. Some elementary algebra shows that we obtain four linearly independent macroscopic state variables:
\begin{equation}\label{eq:strat3-a}
M_1 = \langle X_t^2 \rangle, \qquad M_2 =\left\langle\dfrac{X_t^2}{1- X_t^2/b}\right\rangle-1,	
\end{equation} 
as above, and additionally
\begin{equation}\label{eq:strat3-b}
	M_3 = \left\langle\dfrac{X_t^2}{(1- X_t^2/b)^2}\right\rangle, \qquad
	M_4 = \left\langle\dfrac{X_t^4}{(1- X_t^2/b)^3}\right\rangle. 
\end{equation}
Note that these same macroscopic state variables would also show up after simplification by applying this procedure starting from the choice $M_1=\tau_p$.
If additional moments are desired, one could continue by writing down evolution equations for $M_3$ and $M_4$ and add the moments that appear in those equations, but we will not consider that in the remainder of the text.

\section{A consistency result for FENE-P dumbbells}\label{sec:fene-p}

To check the consistency of the whole procedure, let us apply the numerical closure approximation to the case of FENE-P dumbbells (namely using the spring force \eqref{e:fenep}). In this case, it is known that there exists a macroscopic equivalent model and the question is thus: do we recover this macroscopic model using the numerical closure procedure~? We first derive a theoretical result, which we subsequently illustrate numerically.

\subsection{A simple remark}\label{sec:math}

Let us consider the FENE-P model, with the above numerical closure approximation method applied using only one macroscopic state variable $M=\langle X_t^2 \rangle$. Note that the stress $\tau_p$ is defined in terms of $M$ as $$\tau_p=\frac{\epsilon}{\We} \left(\frac{M}{1-M/b} -1 \right).$$

As mentioned above (see~\eqref{eq:FENE-P}), for the microscopic model~\eqref{eq:1D-preimposed}, $M$ satisfies a closed equation:
\begin{equation}\label{eq:FENE-P_1d}
\partial_t M = 2 \kappa M - \frac{1}{\We} \frac{M}{1-M/b} + \frac{1}{\We}.
\end{equation}

We now make a simple observation to show that the numerical closure approximation (in the limit of zero discretization errors) reproduces this macroscopic dynamics. We refer to the notation of Section~\ref{sec:algo}. For a given value of $M(t^*)$ at time~$t^*$, the lifting step (i) creates an ensemble of configurations with, by construction, a law $\varphi^{NC}_{M(t^*)}=\overline{\mathcal L}(M(t^*))$ such that $\displaystyle \int X^2 \varphi^{NC}_{M(t^*)}(X) \, dX= M(t^*)$. But then, the simulation step (ii) will indeed propagate $M$ according to~\eqref{eq:FENE-P_1d} (which is deduced from~\eqref{eq:1D-preimposed} by a simple It\^o calculus). Thus, after the restriction step (iii), the correct values for $M$ are recovered.

In conclusion, if there exists a closed macroscopic equation for the stress, the proposed numerical closure approximation indeed recovers this macroscopic evolution as soon as the appropriate macroscopic state variables are selected.

\subsection{Numerical illustration}\label{sec:num-fenep}

We consider one-dimensional FENE-P dumbbells, governed by \eqref{eq:1D-preimposed}, in which the spring force $F(X)\equiv F_{FENE-P}(X)$ is given by \eqref{e:fenep} with nondimensional parameters $b=49$, $\We=1$ and $\epsilon=1$. As in \cite{Keunings:1997p9982}, we prescribe the velocity field
\begin{equation}\label{eq:complex_flow}
\kappa(t)=100\; t\; (1-t)\;\exp(-4t).
\end{equation}
 The microscopic model \eqref{eq:1D-preimposed} is discretized via the Euler-Maruyama method with time step $\d t=10^{-2}$. 

\subsubsection{Lifting}

To illustrate that the macroscopic variable $M=\langle X_t^2 \rangle$ uniquely determines the polymer distribution, we perform the following experiment.  We first simulate an ensemble of $N=10^5$ FENE-P dumbbells, subject to the velocity gradient $\kappa(t)$ over the time interval $t\in [0,0.3]$.  As the initial condition, we take the equilibrium polymer distribution in the absence of flow. At $t=0.3$, we obtain $M^*=M(t=0.3)$ via restriction; the corresponding polymer distribution is kept as the reference distribution. Next, we initialize a new ensemble of polymers consistently with the macroscopic state $M^*$ using a uniform distribution. We then perform a constrained simulation \eqref{eq:constrained} using the same time-step $\d t$ over the constrained time interval $[0,m_\infty\delta t]=[0,50]$.  The results are shown in Figure~\ref{fig:fenep_constrained}.  
\begin{figure}
\begin{center}
		\includegraphics{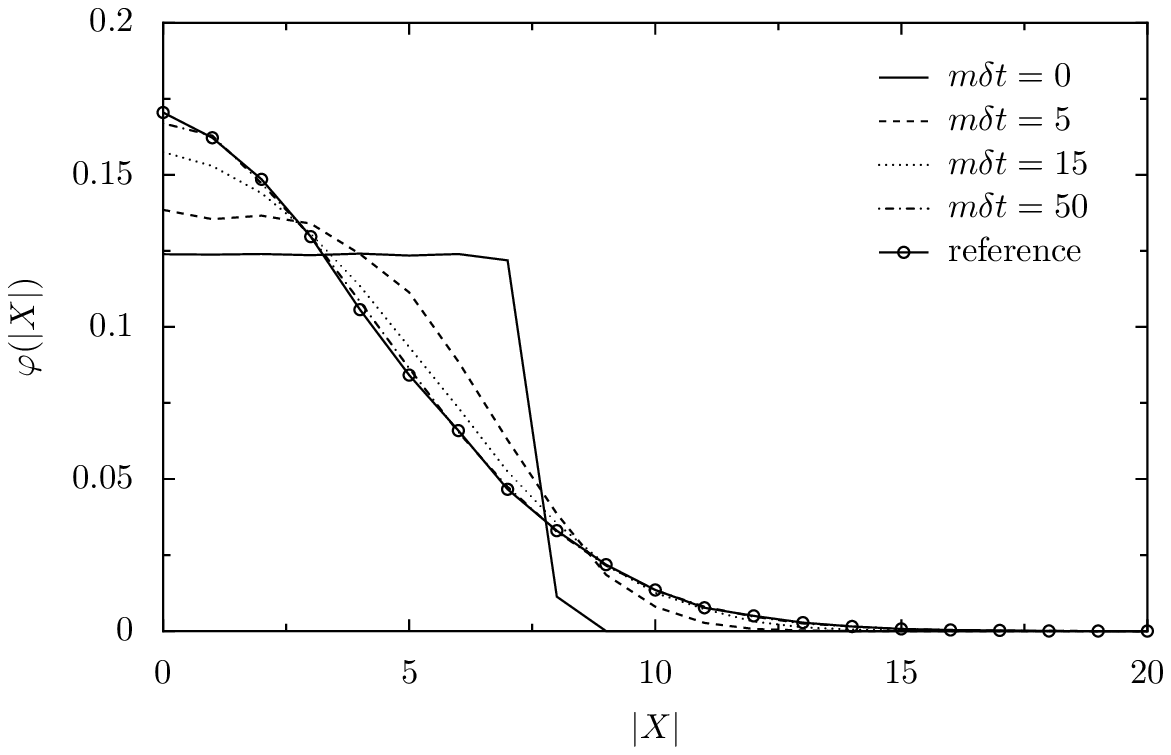}
\end{center}
\caption{\label{fig:fenep_constrained}Polymer distribution for FENE-P dumbbells during constrained simulation.  Shown are the polymer distribution before the restriction at $t=0.3$ (the reference distribution), and at several time instances during a constrained simulation starting from a uniform initial distribution. (The non-uniform appearance of the initial condition is due to artifacts of the binning.) Parameters of the simulation are given in the text.}
\end{figure}

We see that the distribution of the constrained simulation converges towards the distribution of the original simulation, indicating that the first even moment $M$ is indeed sufficient to represent the original polymer distribution, and also that the constrained simulation recovers this distribution. 

Note, however, that this experiment reveals an important property of FENE-P dumbbells.  While the manifold consisting of Gaussian distributions with zero mean is \emph{invariant}, there is \emph{no strong time-scale separation} between the relaxation of arbitrary distributions with given second moment towards the Gaussian distribution and evolution of this second moment itself. This can be concluded by noting that one needs to simulate the constrained SDE over a time interval of length $50$ to reach the stationary distribution, whereas the macroscopic state variable evolves significantly on considerably shorter time-scales, see also the next experiment. This was also observed in \cite{IlgValidity}.

\subsubsection{Coarse time-stepping}

We now look into the evolution of the numerical closure with respect to the full microscopic simulation. To this end, we simulate an ensemble of $N=2\cdot 10^4$ FENE-P dumbbells, starting from the equilibrium distribution $\varphi_{eq}$ in the absence of flow, up to time $t=2$.  All numerical parameters are the same as above. In particular, $\kappa(t)$ is again given by \eqref{eq:complex_flow}. We compare this reference simulation with a number of simulations using the coarse time-stepper with different values of the time step $\Delta t=K\delta t$.  In this experiment, the lifting step amounts to freezing physical time and performing a constrained simulation that is consistent with $M$. The constrained simulations are performed over a time interval of size $100\Delta t$. The results are shown in Figure~\ref{fig:fenep_cts}.
\begin{figure}
\begin{center}
		\includegraphics[scale=0.75]{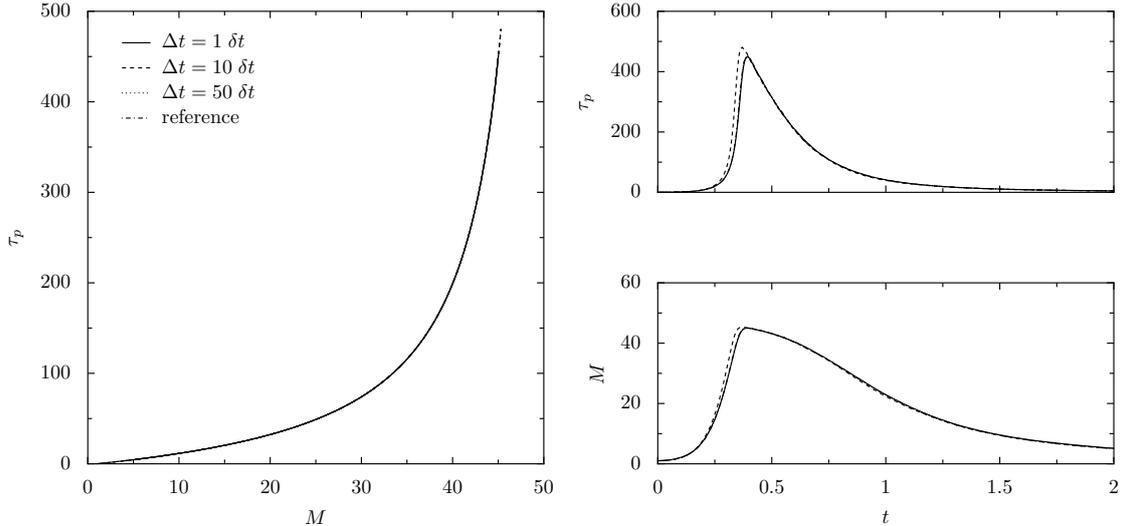}
\end{center}
\caption{\label{fig:fenep_cts}  Evolution of the first even moment $M$ and stress $\tau_p$ for an ensemble of FENE-P dumbbells during complex flow. Left: $(M,\tau_p)$ phase plane view.  Right: temporal evolution. Shown are a full microscopic simulation (reference), and simulations using a coarse time-stepper for different values of the macroscopic time-step.  Simulation parameters are given in the text.}
\end{figure}
We see that the results are nearly identical for all values of $\Delta t$ and the results nearly coincide with the reference simulation. 
% (The differences between different curves are mainly due to the use of different random numbers.) 
This is to be expected.  Indeed, since $M$ completely determines the polymer distribution, a simulation constrained upon $M$ will not alter this distribution, see Section~\ref{sec:math}. 
 % This is merely an illustration that the coarse time-stepper retains the invariant manifold of the microscopic FENE-P model.  \comment{Bof cette derniere phrase}

\section{Comparison of numerical closure with quasi-equilibrium method}\label{sec:eq_ott}

In this section, we compare the proposed numerical closure approximation (described in Section~\ref{sec:ef}) with the quasi-equilibrium method proposed in~\cite{Ilg:2002p10825} (described in Section~\ref{sec:ott}).
In particular, we show that the quasi-equilibrium method, as proposed in \cite{Ilg:2002p10825}, is equivalent to the numerical closure approximation, when the velocity gradient $\kappa(t^*)$ is taken to zero in~\eqref{eq:constrained}. To prove this result, we need to show that the canonical distribution $\varphi^{QE}_\M$ reconstructed from the quasi-equilibrium method (see Equation~\eqref{eq:QE}) is the same as the distribution $\varphi^{NC}_\M$ reconstructed from the lifting procedure through the operator $\bar{\mathcal{L}}$ (see Equations~\eqref{eq:L-inf} and~\eqref{eq:psi_infty}).

Let us consider the microscopic model~\eqref{eq:mic1}--\eqref{eq:mic2}, with a general force $F$ which derives from a potential $\Pi$:
$$F=\nabla \Pi,$$
so that the equilibrium distribution (for zero velocity field) is
$$\varphi_{eq}=Z^{-1} \exp( - \Pi),$$
where $Z= \int \exp( - \Pi )$.
Let us consider a fixed given set of macroscopic state variables $\M$, and, for the sake of simplicity, let us assume that $L=1$ (only one macroscopic state variable $M$ is considered).

From the quasi-equilibrium method, the reconstructed distribution is (see Equation~\eqref{eq:QE}):
\begin{equation}\label{eq:phi_QE}
\varphi^{QE}_M(X)=Z^{QE}_M \exp \left( - \Pi(X) + \lambda m(X) \right),
\end{equation}
where $\displaystyle{Z^{QE}_\M= \int \exp \left( - \Pi(X) + \lambda \, m(X) \right) \, dX}$ and the single Lagrange multiplier $\lambda$ is determined by the constraint $\displaystyle \int m(X) \varphi^{QE}_M(X) \, dX = M$.

Let us now consider the numerical closure approximation described in Section~\ref{sec:ef}, with $\kappa(t^*)=0$ in~\eqref{eq:constrained}. In this case, since $\kappa(t^*)=0$ in~\eqref{eq:XNt}, the stationary distribution for~\eqref{eq:XNt} has a simple expression:
$$\psi^N(\infty,d\mathcal{X}) = (Z^N)^{-1} \prod_{n=1}^N \exp ( - \Pi(X^n) ) d \sigma_{\Sigma(M)},$$
where $\sigma_{\Sigma(M)}$ is the Lebesgue measure on the submanifold $\Sigma(M)$ defined by~\eqref{eq:sigma_M}.
We refer for example to~\cite{Lelievre:2008p9419} or~\cite[Proposition 3.20]{lelievre-rousset-stoltz-10}. Then, the marginal $\psi^N_1(\infty,X)$ is defined through (see~\eqref{eq:psi1N}):
\begin{equation}\label{eq:psi1Ninfty}
\psi^N_1(\infty,X^1)dX^1 = \int_{X^2,\ldots,X^N} \psi^N(\infty,dX^1,\ldots,dX^N),
\end{equation}
and the reconstructed distribution from the numerical closure approximation is (see~\eqref{eq:psi_infty}):
\begin{equation}\label{eq:phi_NC}
\varphi^{NC}_M(X)= \lim_{N \to \infty} \psi^N_1(\infty,X).
\end{equation}
The main mathematical result of this work is the following:
\begin{pro}\label{prop:eq_ott}
The reconstructed distributions obtained through the quasi-equilibrium method, and the numerical closure approximation method with zero gradient velocity field are the same:
$$\varphi^{QE}_M=\varphi^{NC}_M.$$
\end{pro}
This proposition is a corollary of a general result about the equivalence (for an infinite number of particles) of the canonical ensemble and the microcanonical ensemble in statistical physics. We cite a result from~\cite[Theorem A.5.5]{bernardin-olla-10}, see also~\cite[Theorem 3.4]{guo-papanicolaou-varadhan-88}:
\begin{thm}\label{theo:eq_ens}
Let $\alpha$ be a probability measure on $\RR^d$ and let us consider $Y^1,\ldots,Y^N$ i.i.d. random variables with law $\alpha$, and introduce a function  $q: \RR^d \to \RR$. Let us now define two probability measures: 
\begin{itemize}
\item The conditional measure
$$\nu^{N}_{|z}\left(dy^1,\ldots,dy^N\right)=\alpha^{\otimes N} \left(dy^1,\ldots,dy^N\Bigg| \dfrac{1}{N} \sum_{n=1}^N q(y_n)=z\right)$$
of the vector $(Y^1,...,Y^N)$ conditionally to $\dfrac{1}{N} \sum_{n=1}^N q(Y^n)=z$.
\item The probability measure
$$\alpha_{\lambda}(dy) = Z_{\lambda}^{-1} \exp( \lambda q (y)) \, \alpha(dy),$$
\end{itemize}
where $Z_{\lambda}=\int \exp( \lambda q (y)) \, \alpha(dy)$. Let us assume that $\lambda$ and $z$ are related through the relation:
$$\int q(y) \alpha_{\lambda}(dy) = z.$$
Then, one has: for any test function $F: \RR^d \to \RR$,
$$\lim_{N \to \infty} \int F(y^1) \, \nu^N_{|z}\left(dy^1,\ldots,dy^N\right) =  \int F(y^1) \alpha_\lambda(dy^1).$$
\end{thm}
To apply Theorem~\ref{theo:eq_ens} to prove Proposition~\ref{prop:eq_ott}, we set $\alpha$ to be the equilibrium distribution $\varphi_{eq}$, $q=m$, and $z=M$. Then $\alpha_\lambda=\varphi^{QE}_M$, and it remains to show that
$$\lim_{N \to \infty} \int F(y^1) \, \nu^N_{|z}\left(dy^1,\ldots,dy^N\right) = \int F(y^1) \varphi^{NC}_M(y^1) \, dy^1.$$
This is stated in the following lemma:
\begin{lem}
Let us consider the notation of Theorem~\ref{theo:eq_ens} and assume that the measure $\alpha$ has a density $a$: $$\alpha(dy)=a(y)\,dy.$$ Let us introduce the probability measure
$$\nu^N_{\Sigma(z)} (dy^1, \ldots, dy^N) = a(y^1)\cdots a(y^N) \sigma_{\Sigma^N(z)} (dy^1, \ldots, dy^N),$$
where $\Sigma^N(z)=\{(y^1, \ldots, y^N), \, \frac{1}{N} \sum_{n=1}^N q(y^n)=z \}$ and $\sigma_{\Sigma^N(z)}$ is the Lebesgue measure on the submanifold $\Sigma^N(z)$. Then,
\begin{equation}\label{eq:co_area}
\nu^N_{\Sigma^N(z)} (dy^1, \ldots, dy^N) = \|\nabla Q^N \| \, \nu^{N}_{|z}\left(dy^1,\ldots,dy^N\right),
\end{equation}
where $Q^N(y^1, \ldots, y^N)=\frac{1}{N} \sum_{n=1}^N q(y^n)$. Moreover,
\begin{equation}\label{eq:lim}
\lim_{N \to \infty} \int F(y^1) \, \nu^N_{|z}\left(dy^1,\ldots,dy^N\right) =\lim_{N \to \infty} \int F(y^1) \, \nu^N_{\Sigma^N(z)}\left(dy^1,\ldots,dy^N\right).
\end{equation}
\end{lem}
\begin{proof}
The proof of~\eqref{eq:co_area} is based on the co-area formula, see for example~\cite[Eq. (3.14)]{lelievre-rousset-stoltz-10}. Then, to prove~\eqref{eq:lim}, one notice that, if $Y^1,\ldots Y^N$ denotes random variables distributed according to the conditional probability measure $\nu^N_{|z}$, one has:
\begin{align*}
\int F(y^1) \, \nu^N_{\Sigma^N(z)}\left(dy^1,\ldots,dy^N\right)
&= \frac{\left\langle F(Y^1) \|\nabla Q^N \| (Y^1, \ldots , Y^N) \right\rangle}{\left\langle \|\nabla Q^N \| (Y^1, \ldots , Y^N) \right\rangle} \\
&= \frac{\left\langle F(Y^1) \sqrt{\frac{1}{N} \sum_{n=1}^N \|\nabla q \|^2 (Y^n)} \right\rangle}{\left\langle \sqrt{\frac{1}{N} \sum_{n=1}^N \|\nabla q \|^2 (Y^n)} \right\rangle}.
\end{align*}
 By a law of large numbers (see for example~~\cite[Theorem A.5.4]{bernardin-olla-10} or~\cite[Theorem 3.5]{guo-papanicolaou-varadhan-88}), $\displaystyle \frac{1}{N} \sum_{n=1}^N \|\nabla q \|^2 (Y^n)$ converges in probability to $\displaystyle \int \|\nabla q \|^2 d \alpha_{\lambda}$, and thus, Slutsky lemma enables to conclude. 
\end{proof}
This concludes the proof of Proposition~\ref{prop:eq_ott}, since with the notation introduced above ($\alpha(dy)=\varphi_{eq}(y) dy$, $q=m$, and $z=M$) $$\nu^N_{\Sigma^N(z)} (dy^1, \ldots, dy^N)=\psi^N(\infty,dy^1, \ldots, dy^N).$$

A few remarks are in order. First, in dimension 1, the fact that the drift in the SDE derives from a potential is not a restrictive assumption, so that the quasi-equilibrium procedure could also be applied when taking into account a non-zero $\kappa(t^*)$. However, this assumption is indeed restrictive in dimension greater than one: for non-symmetric ${\boldsymbol \kappa}(t^*)$, the drift in~\eqref{eq:mic2_lag} is not the gradient of a potential. In this case, the numerical closure approximation procedure still applies, but it is unclear how it would be related to a quasi-equilibrium method. In some sense, the numerical closure method can thus be seen as a generalization of the quasi-equilibrium method, which takes into account the velocity gradient in the lifting procedure. In fact, the numerical closure procedure can be seen as a simple alternative to simulate the quasi-equilibrium closures that, unlike the numerical procedure in \cite{Ilg:2002p10825}, does not require transformations from moments to Lagrange multipliers and vice versa, which might be difficult to perform.

\section{Numerical illustrations for FENE dumbbells}\label{sec:num-fene}

In this section, we perform some numerical experiments to explore the behaviour of the numerical closure procedure using the strategies for macroscopic state variable detection that were outlined in Section~\ref{sec:strat}.

\subsection{Strategy 1: Even moments as macroscopic state variables}

\subsubsection{Lifted configuration distributions}

 We simulate an ensemble of $N=5\cdot 10^4$ FENE dumbbells, subject to a constant velocity gradient $\kappa(t)=2$ over the time interval $t\in [0,t^*]$, with $t^*=0.5,1,1.5,2$ (startup of ``elongational'' flow). We use nondimensional parameters $b=49$ and $\We=\epsilon=1$, and choose $\delta t=2\cdot 10^{-4}$. As the initial condition, we take the equilibrium polymer distribution in the absence of flow. As the macroscopic state variables, we take the first $L$ even moments. At $t=t^*$, we obtain $\M^*=\mathcal{R}(\mathcal{X}^*)$ via restriction; the corresponding polymer distribution is kept as the reference distribution. Starting from $\mathcal{X}^*$, we then perform a constrained simulation under the constraint that $\mathcal{R}(\mathcal{X})=\M^*$, using the same time-step $\d t$, until the polymer distribution equilibrates.   Figure~\ref{fig:fene_nblin} shows the constrained equilibrium polymer distributions for a range of values of $L$.   
\begin{figure}
\begin{center}
	\includegraphics[scale=0.6]{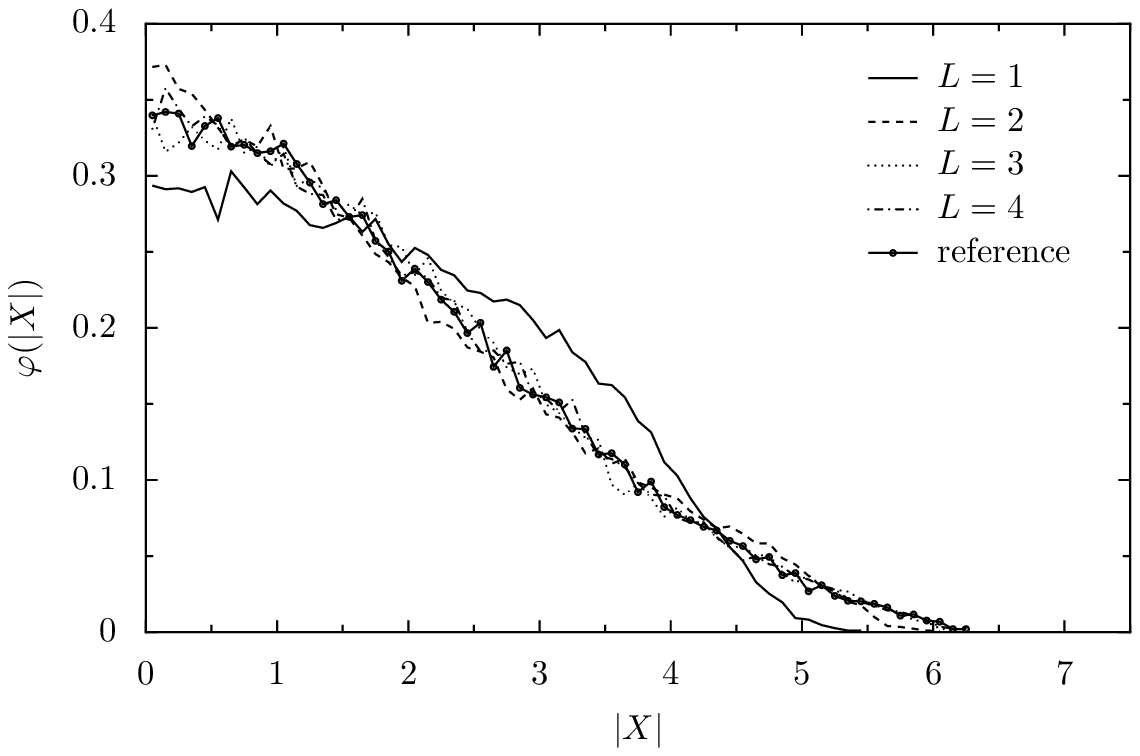}
	\includegraphics[scale=0.6]{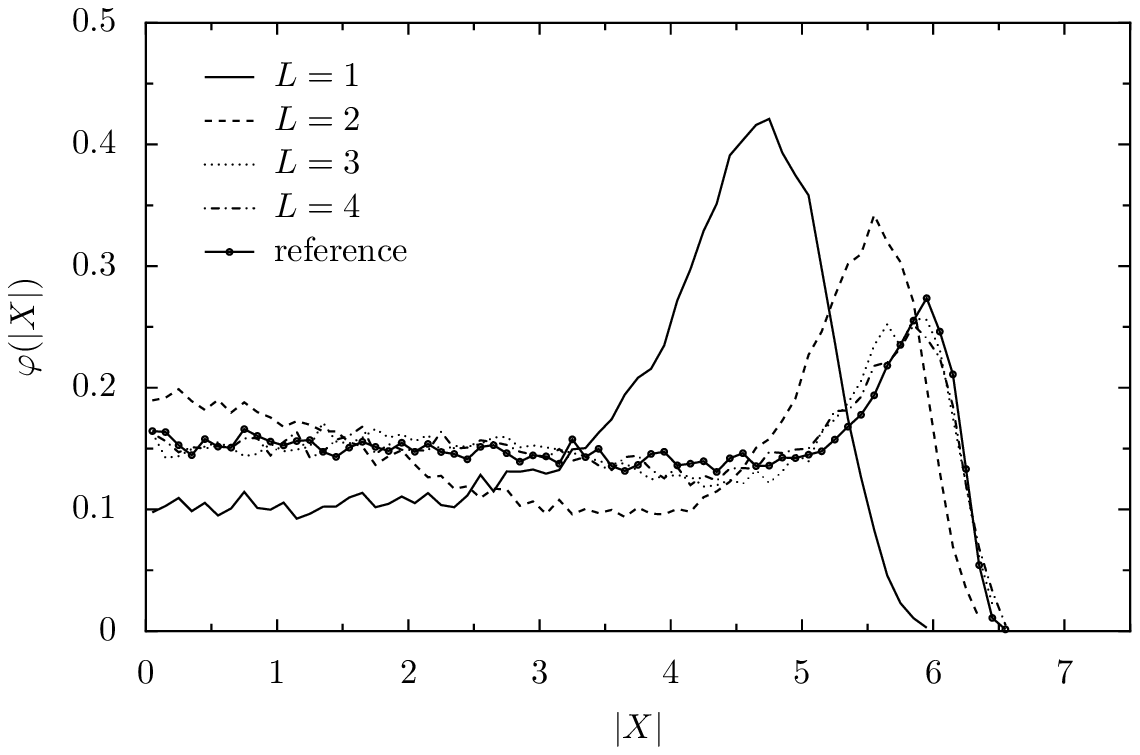}
	\includegraphics[scale=0.6]{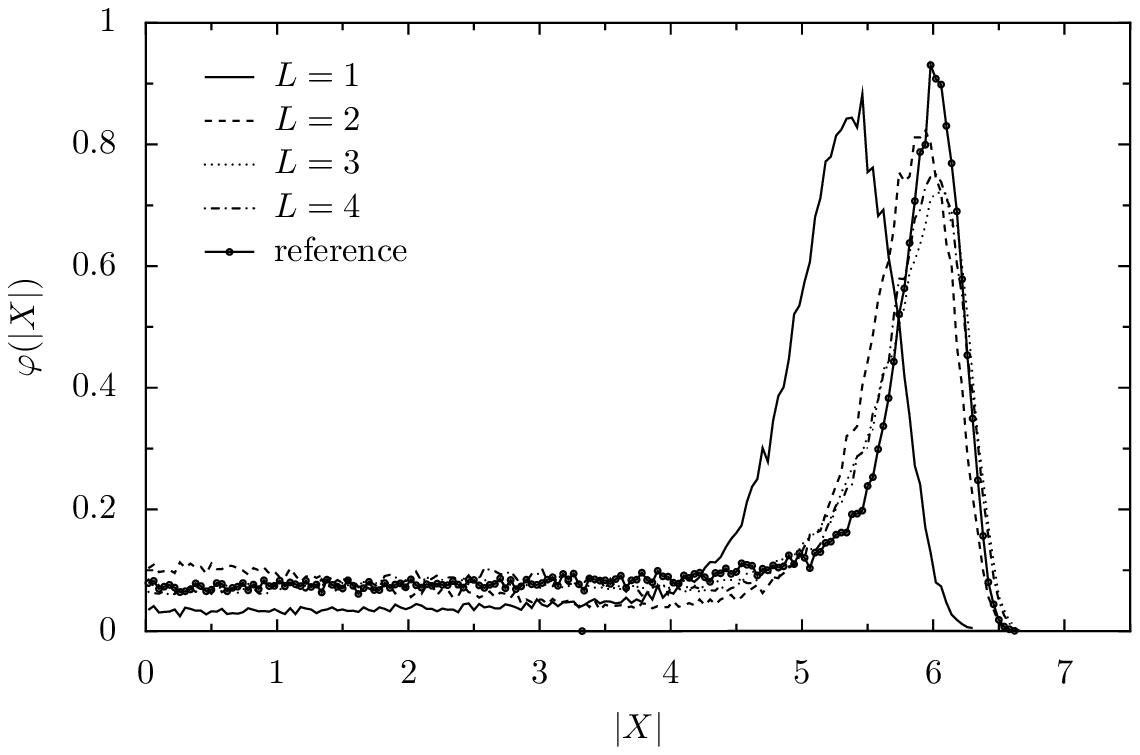}
	\includegraphics[scale=0.6]{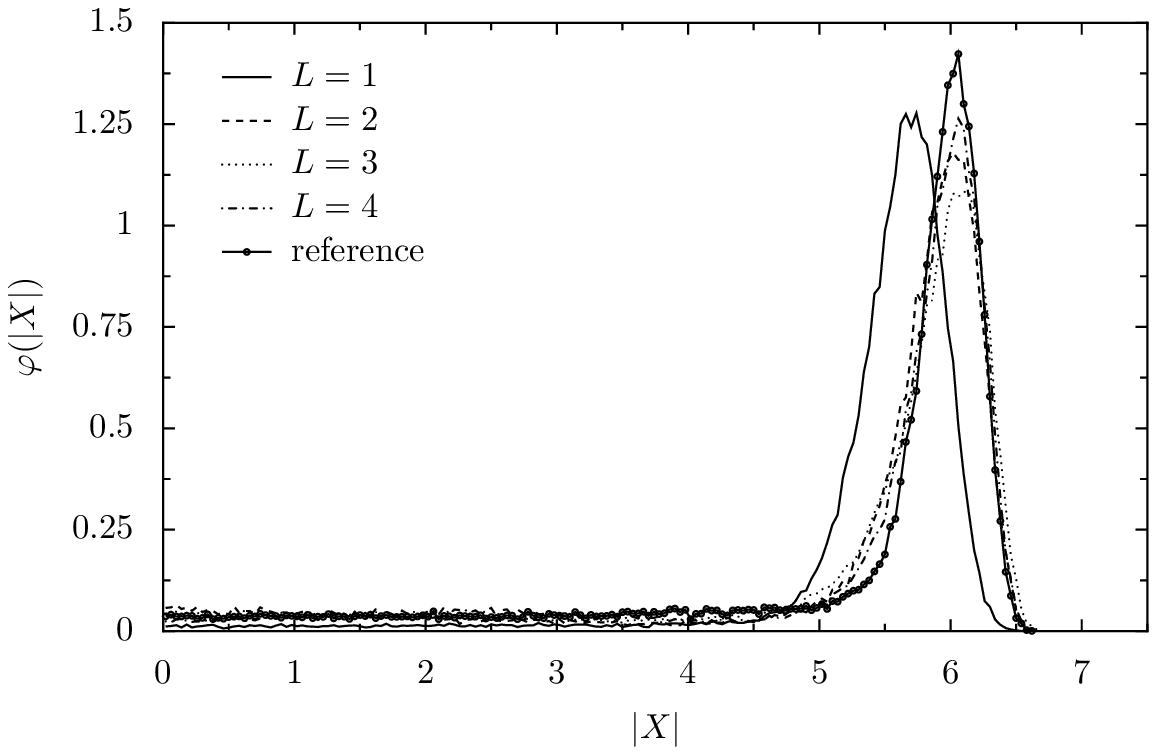}
\end{center}
\caption{\label{fig:fene_nblin} Lifted polymer distributions for FENE dumbbells as a function of the number of macroscopic state variables using strategy 1.  We plot a reference polymer distribution, that is obtained by microscopic simulation up to time $t^*$, as well as the equilibrium polymer distributions after constrained simulation using $L=1,\ldots,4$ even moments. Shown are the results for $t^*=0.5$ (top left), $t^*=1$ (top right), $t^*=1.5$ (bottom left) and $t^*=2$ (bottom right). Simulation parameters are given in the text.}
\end{figure}
We see that, as the number of macroscopic state variables increases, the difference decreases between the constrained equilibrium distribution and the reference distribution, indicating that this distribution is captured more accurately when more macroscopic state variables are used. 

\subsubsection{Relaxation to equilibrium and comparison with quasi-equilibrium approach\label{sec:relax}}

We now repeat the above experiment with $N=2000$ particles and $t^*=1$, and plot the evolution of the polymer stress $\tau_p$ as a function of time.  All other simulation parameters are as above. Moreover, to obtain the corresponding result for the quasi-equilibrium method of \cite{Ilg:2002p10825}, we perform the same experiment, but now with $\kappa(t)=0$ throughout the constrained simulations.  We ensured that both constrained simulations were performed using the same random numbers.  The results are shown in Figure~\ref{fig:ottinger}.
\begin{figure}
\begin{center}
	\includegraphics[scale=0.7]{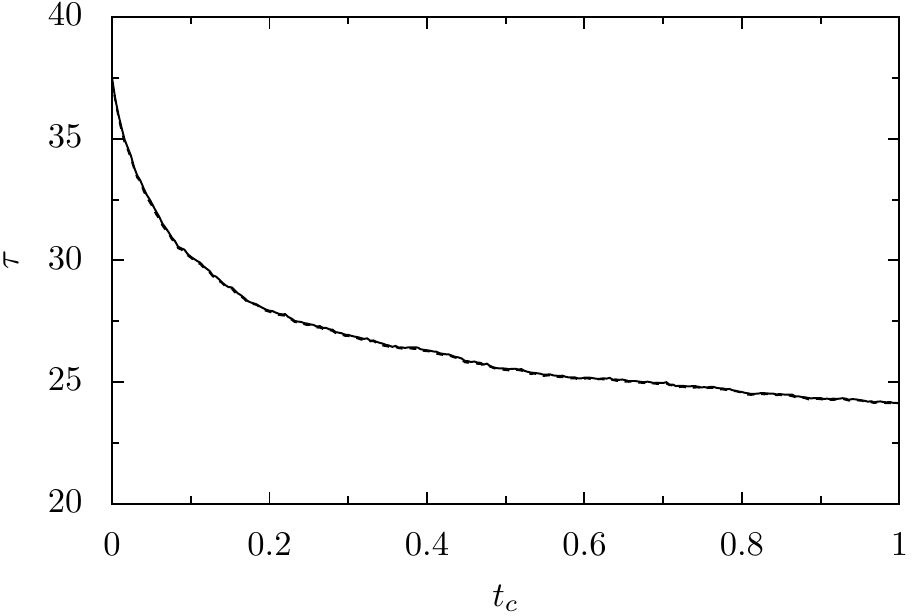}
	\includegraphics[scale=0.7]{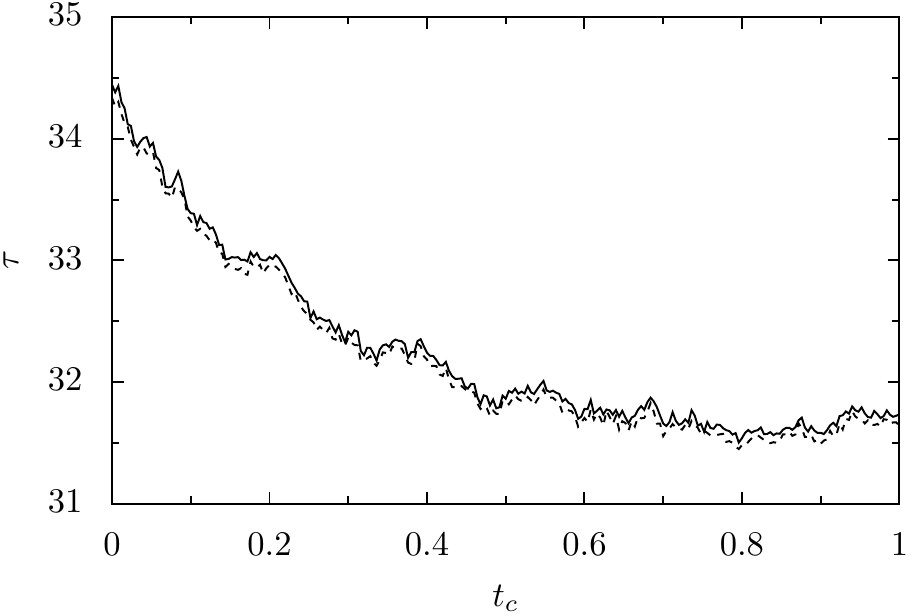}
	\includegraphics[scale=0.7]{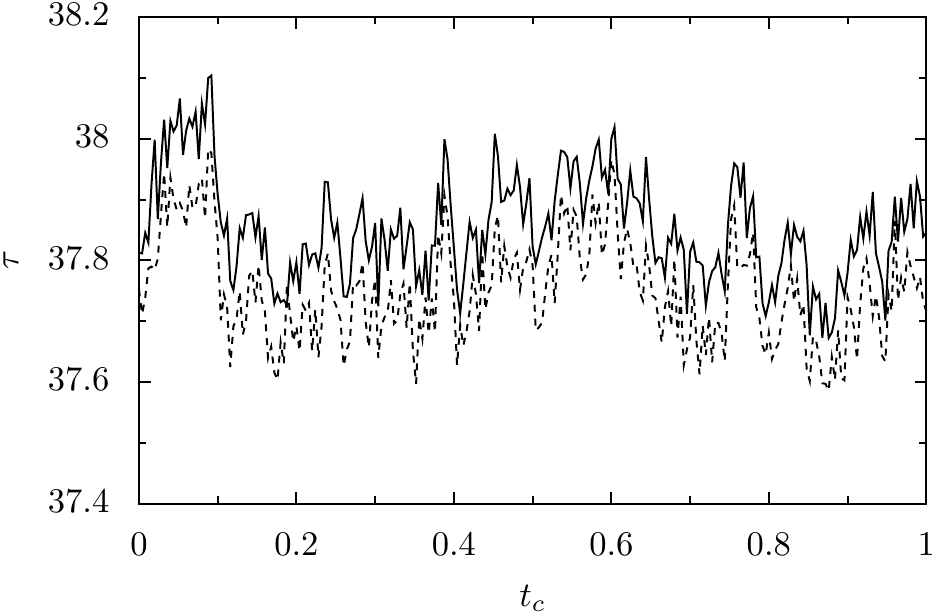}
	\includegraphics[scale=0.7]{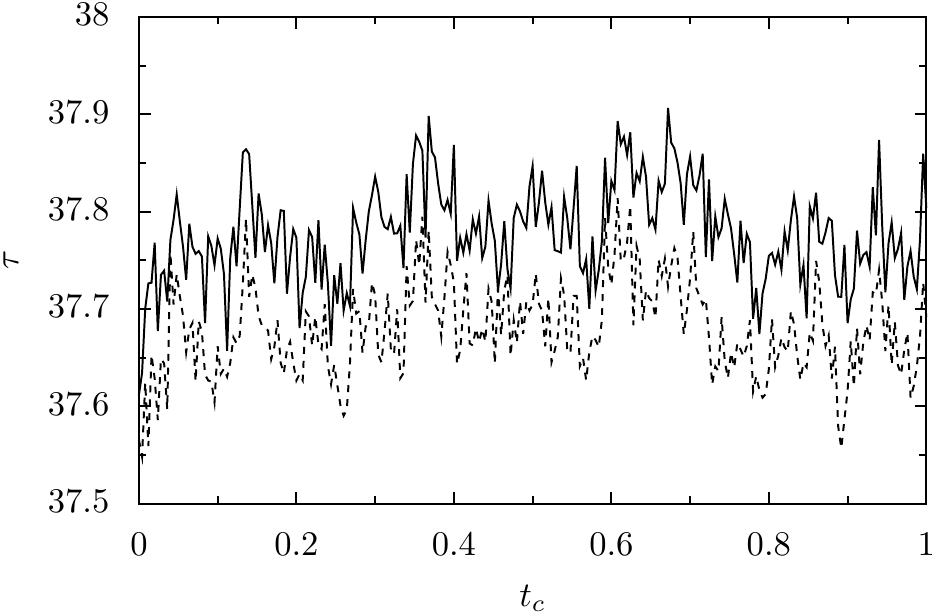}
\end{center}
\caption{\label{fig:ottinger} Evolution of the stress tensor $\tau_p$ throughout constrained simulation using strategy 1 with $L=1$ (top left), $L=2$ (top right), $L=3$ (bottom left) and $L=4$ (bottom right). Solid lines are obtained using the procedure outlined in Section~\ref{sec:ef}; dashed lines correspond to the quasi-equilibrium approximation. Simulation parameters are given in the text.}
\end{figure}
The figures clearly show a relaxation towards the stress value that corresponds to the lifted polymer distribution. This fact can be used to detect when the constrained simulation has equilibrated, and hence to determine the parameter $m_\infty$ that was introduced when defining the lifting operator in Section~\ref{sec:ef}. When using the other strategies to determine the hierarchy of macroscopic state variables, $\tau_p$ belongs to the set of macroscopic state variables, and therefore does not change during relaxation. However, in similar experiments, not reported here, we observed similar behaviour when monitoring the first even moment that was not constrained.

Moreover, when the number of macroscopic state variables increases, the stress $\tau_p$ that corresponds to the lifted distribution approaches the stress associated with the distribution that corresponds to the initial condition of the constrained simulation.  This observation is in agreement with the previous experiment, where we showed that the distributions themselves approach the initial distribution of the constrained simulation when more moments are taken into account.  Hence, monitoring the evolution of $\tau_p$ during constrained simulation can be used to determine whether the currently used set of macroscopic state variables is sufficient.  Finally, concerning the relation between the numerical closure and the quasi-equilibrium approximation, we see that the difference between the two approaches is not really large; however, this difference remains of the same order of magnitude, independently of the number of macroscopic state variables included.  

\subsubsection{Coarse time-stepping}

We now look into the evolution of the numerical closure with respect to the full microscopic simulation, again using $\kappa(t)=2$. To this end, we simulate an ensemble of $N=2000$ FENE dumbbells, starting from the equilibrium distribution in the absence of flow, up to time $t=4$.  All  parameters are the same as above.  We compare this reference simulation with a simulation via the coarse time-stepper, using a range of values for the number $L$ of macroscopic state variables; here, the macroscopic time-step is equal to one microscopic step $\delta t$, {\em i.e.}\ $K=1$.  In this experiment, the lifting step amounts to freezing physical time and performing a constrained simulation that is consistent with $\M$. The constrained simulations are performed until equilibrium of the distribution is reached (here using $m_\infty=50$ constrained time steps of size $\delta t$); all simulations were verified to have converged with respect to the number of constrained time-steps.  The results are shown in Figure~\ref{fig:nblin_cts_startup_notau}.
\begin{figure}
\begin{center}	\includegraphics[scale=0.75]{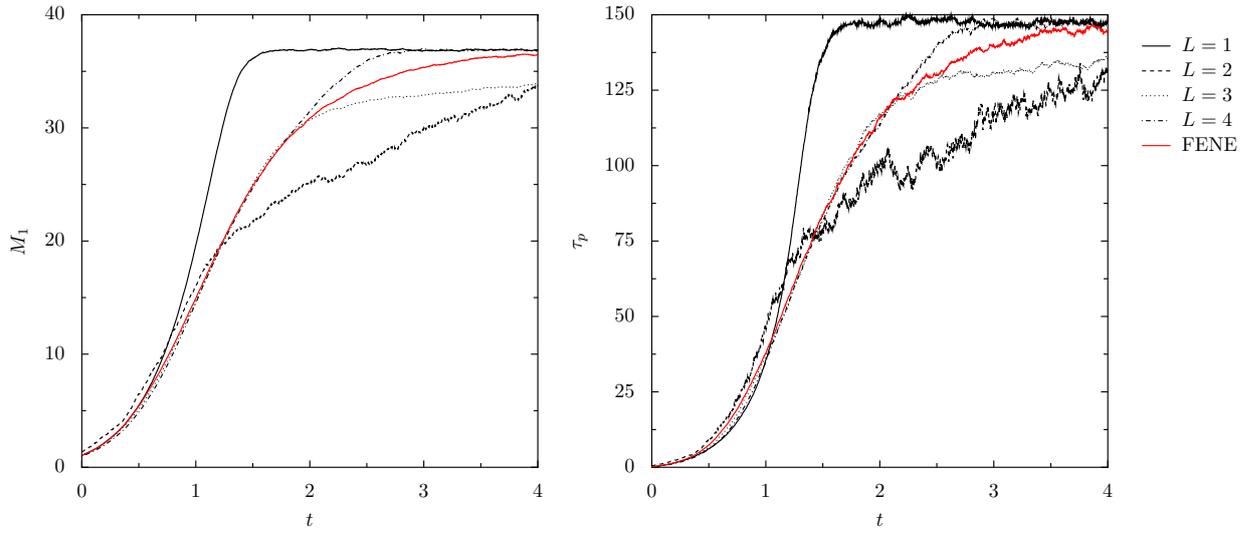}
\end{center}
\caption{\label{fig:nblin_cts_startup_notau}  Evolution of first even moment $M_1$ (left) and stress $\tau_p$ (right) for an ensemble of FENE dumbbells during startup of elongational flow. Shown are a full microscopic simulation (reference), and simulations using a coarse time-stepper for different numbers $L$ macroscopic state variables using strategy 1.  Simulation parameters are given in the text.}
\end{figure}
We clearly see that the approximation improves as a function of the number of moments that are included at the macroscopic level.  Other experiments, not reported here, indicate that the higher $\kappa(t)$, the higher the number of macroscopic state variables that needs to be considered. These results are in line with the conclusions in \cite{Ilg:2002p10825}, where analytical (quasi-equilibrium) closures were obtained via an entropy maximization principle.

Finally, we consider an ensemble of $N=2000$ FENE dumbbells subject to the time-dependent flow field \ref{eq:complex_flow}, and again look at a coarse time-stepper in which the macroscopic state is represented with an increasing number of even moments. For this test, $m_\infty=100$; all remaining simulation parameters are as above.  The results are shown in Figure~\ref{fig:fene_cts_notau}.
\begin{figure}
\begin{center}
		\includegraphics[scale=0.75]{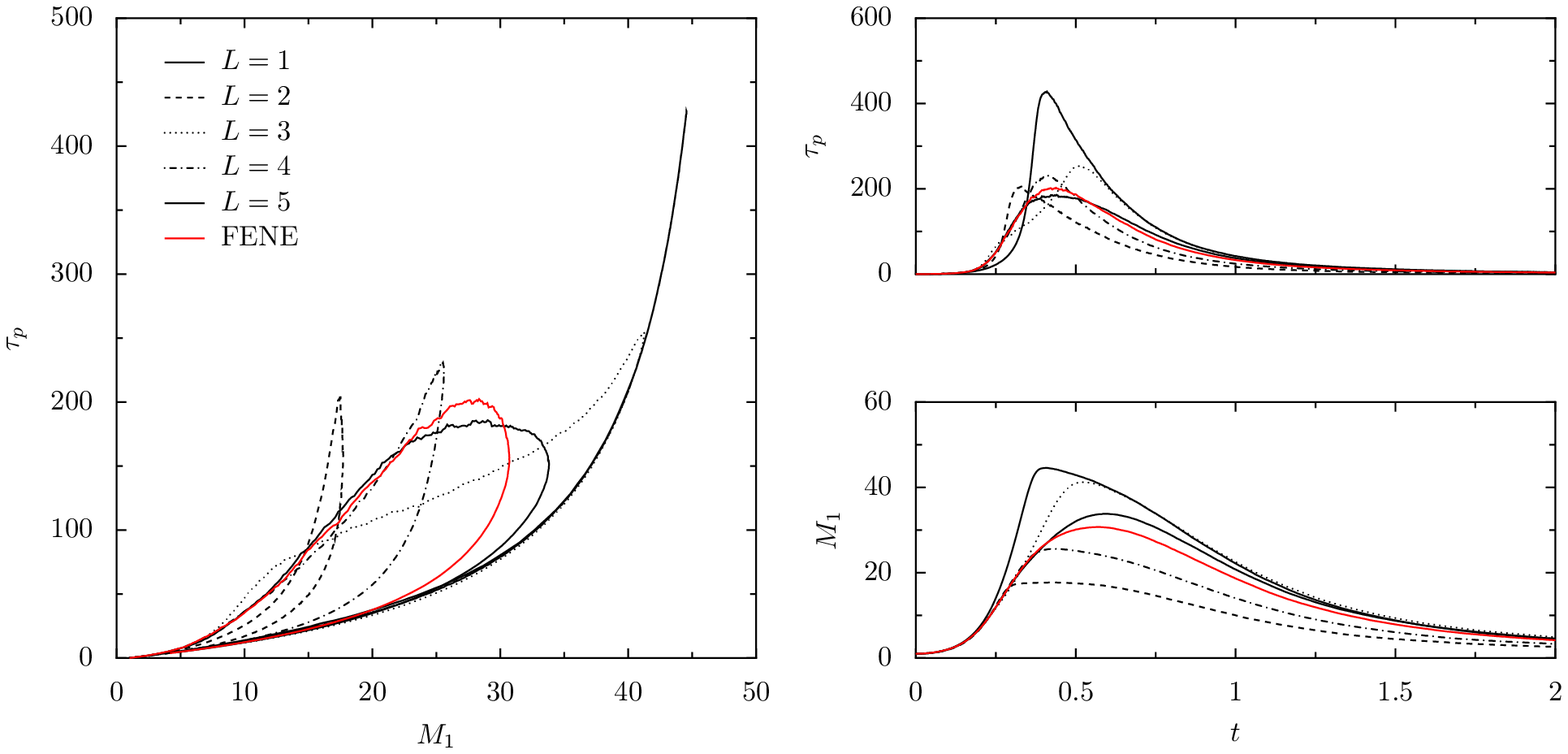}
\end{center}
\caption{\label{fig:fene_cts_notau}  Evolution of first even moment $M_1$ and stress $\tau_p$ for an ensemble of FENE dumbbells during complex flow. Left: $(M_1,\tau_p)$ phase plane view.  Right: temporal evolution. Shown are a full microscopic simulation (reference), and simulations using a coarse time-stepper for different numbers of macroscopic state variables using strategy 1. Simulation parameters are given in the text.}
\end{figure}
The conclusions for this experiment are similar.  Note that a macroscopic description with only one moment cannot capture the hysteretic effect of the FENE dumbbells.

\subsection{Strategy 2: Adding the stress tensor as a macroscopic variable\label{sec:num-fene2}}

One particular advantage of the numerical closure strategy described here is that one can readily consider the effect of considering more complicated moments in the set of macroscopic state variables. In this section, we repeat the above experiments, now considering the first $L-1$ even moments, supplemented with the stress $\tau_p$ itself as a macroscopic variable, {\em i.e.}, $\M=(M_l)_{l=1}^{L}$ with $M_l=\langle X^{2l}\rangle$ for $1\le l \le L-1$, as before, and  $M_L=\tau_p$. 

\subsection{Lifted configuration distributions}

We again simulate an ensemble of $N=5\cdot 10^4$ FENE dumbbells, subject to a constant velocity gradient $\kappa(t)=2$ over the time interval $t\in [0,t^*]$, with $t^*=0.5,1,1.5,2$ (startup of elongational flow) and obtain $\M^*=\mathcal{R}(\mathcal{X}^*)$ via restriction; the corresponding polymer distribution is kept as a reference distribution. We perform a constrained simulation, starting from $\mathcal{X}^*$, under the constraint that $\mathcal{R}(\mathcal{X})=\M^*$ using the same time-step $\d t$, until the polymer distribution equilibrates. Figure \ref{fig:fene_nblin_tau} shows the constrained equilibrium polymer distributions for a range of values of $L$.   
\begin{figure}
\begin{center}
	\includegraphics[scale=0.6]{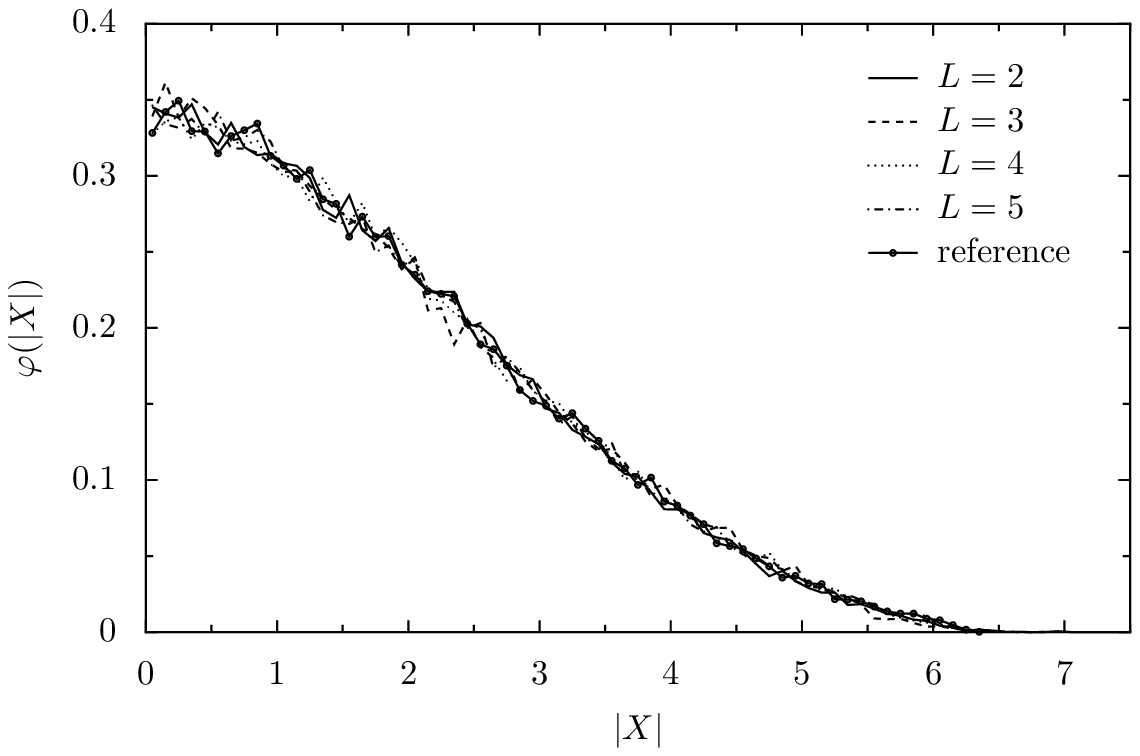}
	\includegraphics[scale=0.6]{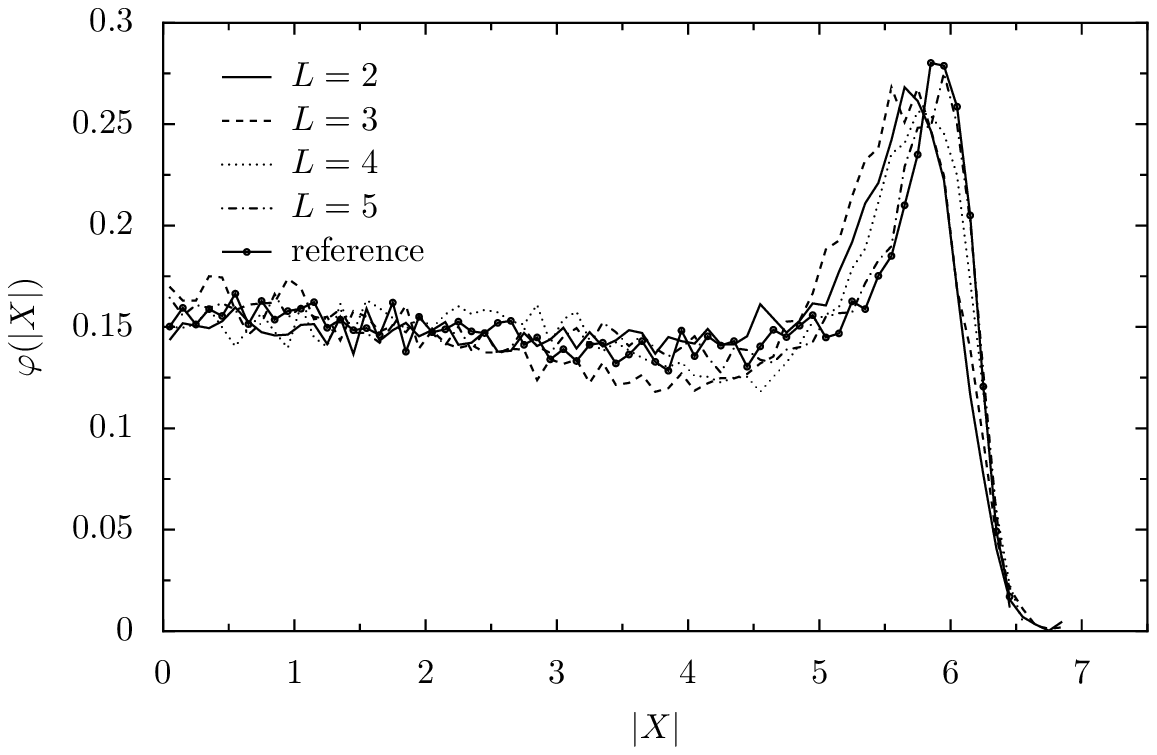}
	\includegraphics[scale=0.6]{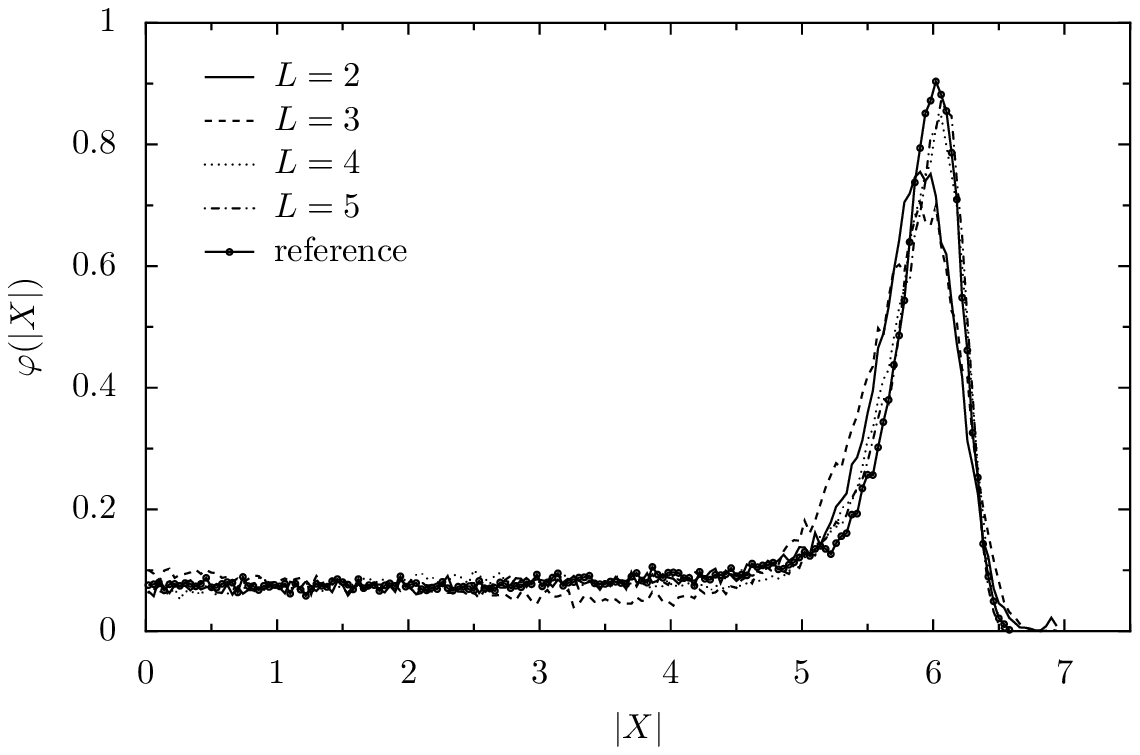}
	\includegraphics[scale=0.6]{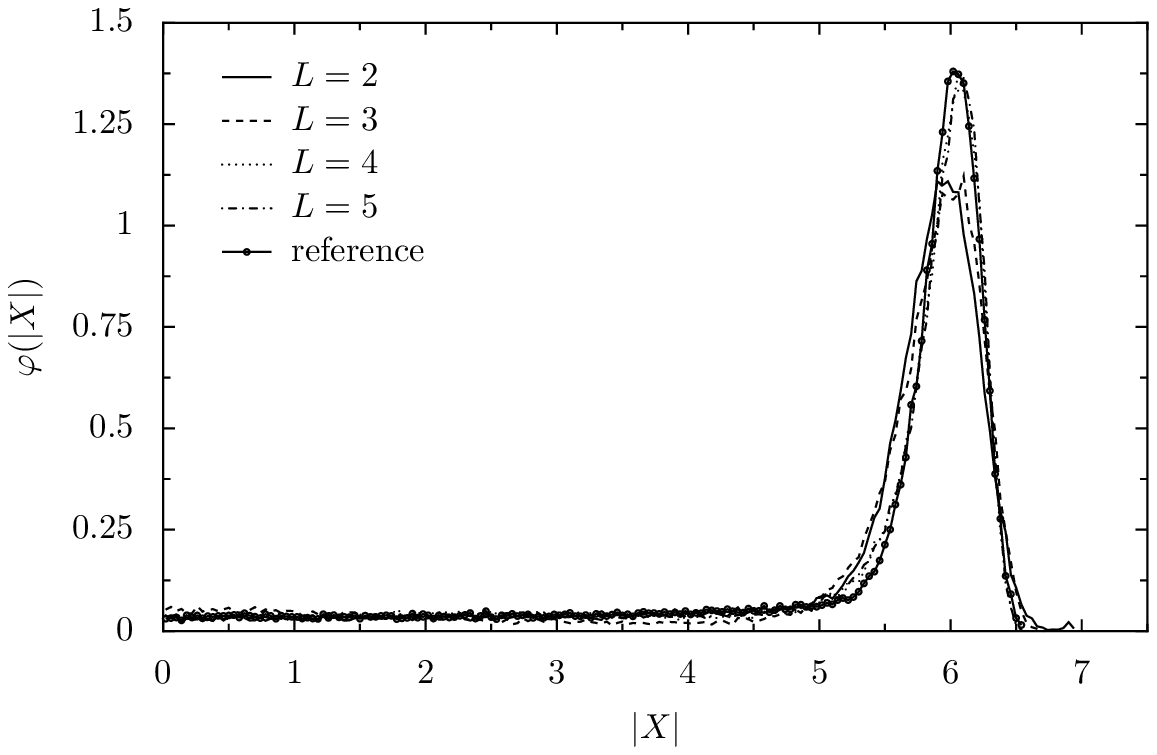}
\end{center}
\caption{\label{fig:fene_nblin_tau} Lifted polymer distributions as a function of the number of macroscopic state variables.  We plot a reference distribution, {\em i.e.}, the polymer distribution after a microscopic simulation up to time $t^*$, as well as the equilibrium polymer distributions after constrained simulation with $L=2,\ldots,5$ moments using strategy 2. Shown are the results for $t^*=0.5$ (top left), $t^*=1$ (top right), $t^*=1.5$ (bottom left) and $t^*=2$ (bottom right). Simulation parameters are given in the text.}
\end{figure}
Compared to the case when only even moments were used, we see that adding $\tau_p$ as a macroscopic variable dramatically improves the obtained equilibrium distributions, and less moments may suffice to characterize the distributions. However, when $L=2$ and $L=3$, we see a peculiar artifact in the distributions, in the sense that we obtain an increase of the number of polymers with near-maximal length (a small second peak in the distributions on the right).  This results in high probability of rejections throughout the constrained simulation.

\subsubsection{Coarse time-stepping}

We now look at the evolution of the numerical closure with respect to the full microscopic simulation, again using $\kappa(t)=2$.  We simulate an ensemble of $N=2000$ FENE dumbbells, starting from the equilibrium distribution in the absence of flow, up to time $t=4$ and compare this reference simulation with a number of simulations using the coarse time-stepper with a different number $p$ macroscopic state variables ($L-1$ even moments, supplemented with the stress tensor $\tau_p$). As before, we choose the macroscopic time-step equal to one microscopic step $\delta t$, {\em i.e.}, $K=1$; all other parameters are also chosen as above.  We lift by freezing physical time and performing a constrained simulation that is consistent with $\M$ until equilibrium of the distribution is reached (here using $m_\infty=50$ constrained time-steps of size $\delta t$); all simulations were verified to have converged with respect to the number of constrained time-steps.  The results are shown in Figure~\ref{fig:nblin_cts_startup_notau}.
\begin{figure}
\begin{center}	\includegraphics[scale=0.75]{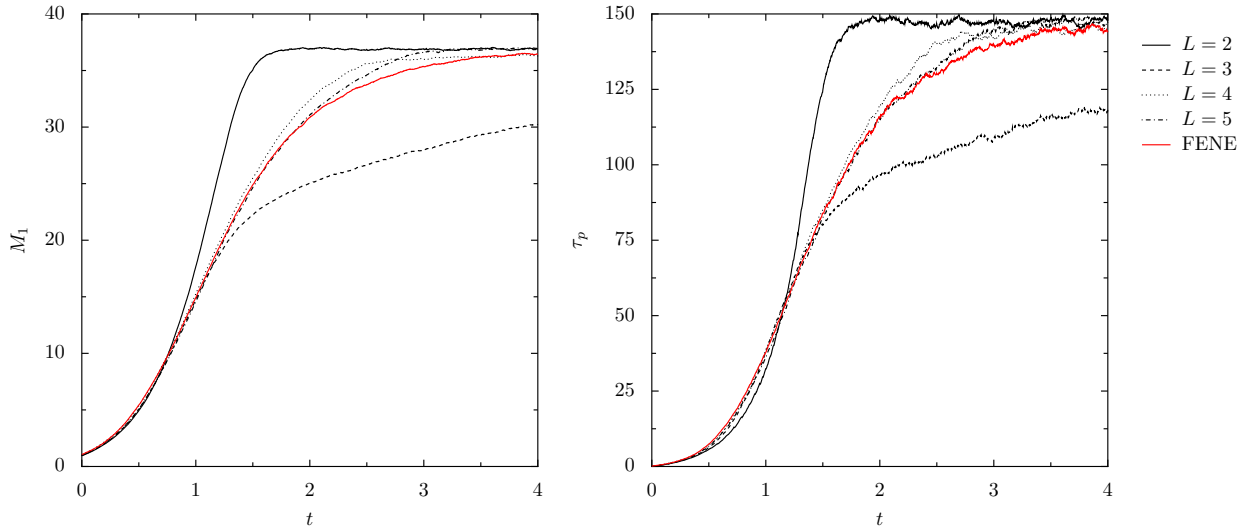}
\end{center}
\caption{\label{fig:nblin_cts_startup_tau}  Evolution of the first even moment $M_1$ (left) and stress $\tau_p$ (right) for an ensemble of FENE dumbbells during startup of elongational flow. Shown are a full microscopic simulation (reference), and simulations using a coarse time-stepper for different numbers of macroscopic state variables  using strategy 2.  Simulation parameters are given in the text.}
\end{figure}
Also here, we see an improvement; the result of the complex flow experiment is shown in figure \ref{fig:fene_cts_tau}.
\begin{figure}
\begin{center}
		\includegraphics[scale=0.75]{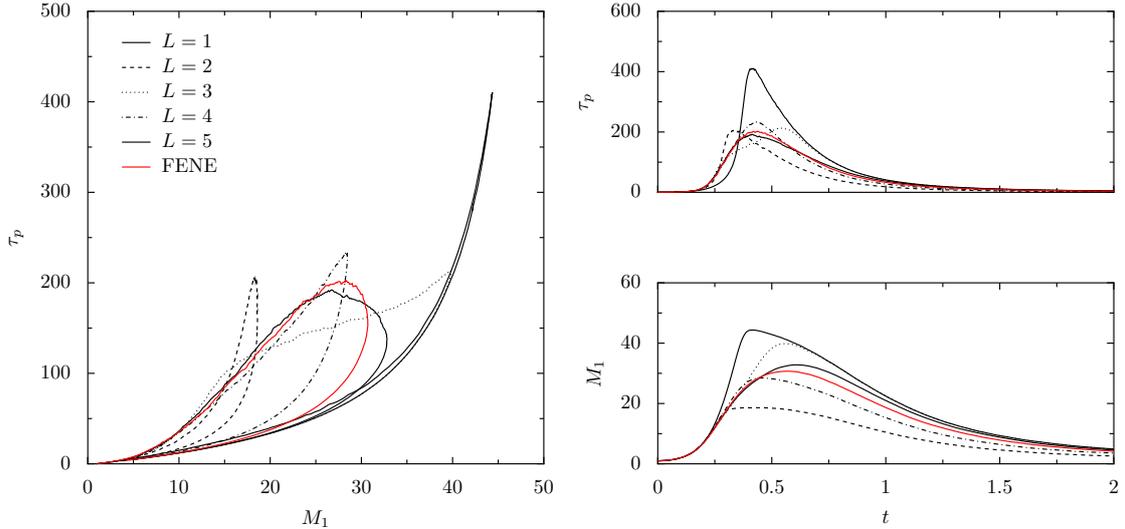}
\end{center}
\caption{\label{fig:fene_cts_tau}  Evolution of first even moment $M_1$ and stress $\tau_p$ for an ensemble of FENE dumbbells during complex flow. Left: $(M_1,\tau_p)$ phase plane view.  Right: temporal evolution. Shown are a full microscopic simulation (reference), and simulations using a coarse time-stepper for different numbers of macroscopic state variables using strategy 2.  Simulation parameters are given in the text.}
\end{figure}

\subsection{Strategy 3: Cascading from the equation of state for $\tau_p$}

Finally, we repeat the above experiments, now considering the moments to be determined by Strategy 3. 

\subsubsection{Lifted configuration distributions}

We again simulate an ensemble of $N=5\cdot 10^4$ FENE dumbbells, subject to a constant velocity gradient $\kappa(t)=2$ over the time interval $t\in [0,t^*]$, with $t^*=0.5,1,1.5,2$ (startup of elongational flow) and obtain $\M^*=\mathcal{R}(\mathcal{X}^*)$ via restriction; the corresponding polymer distribution is taken as the reference distribution. We perform a constrained simulation, starting from $\mathcal{X}^*$, under the constraint that $\mathcal{R}(\mathcal{X})=\M^*$ using the same time-step $\d t$, until the polymer distribution equilibrates. Figure~\ref{fig:stressmu_constrained} shows the constrained equilibrium polymer distributions for an increasing number of macroscopic state variables.   
\begin{figure}
\begin{center}
	\includegraphics[scale=0.6]{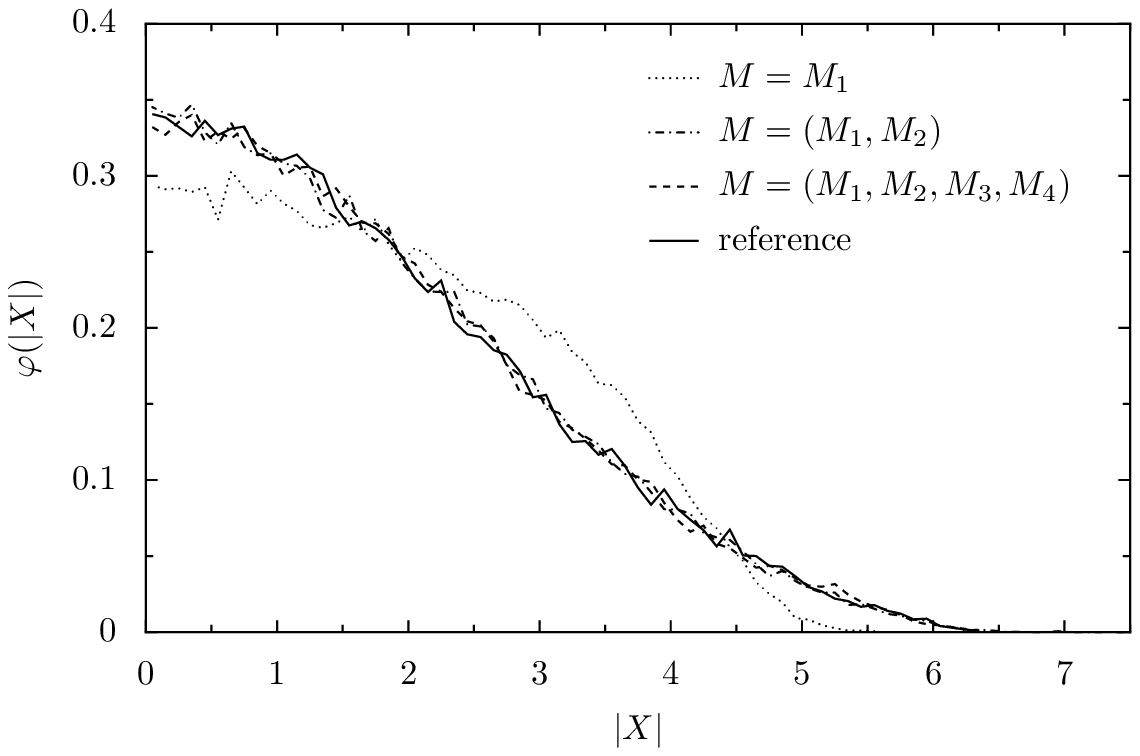}
	\includegraphics[scale=0.6]{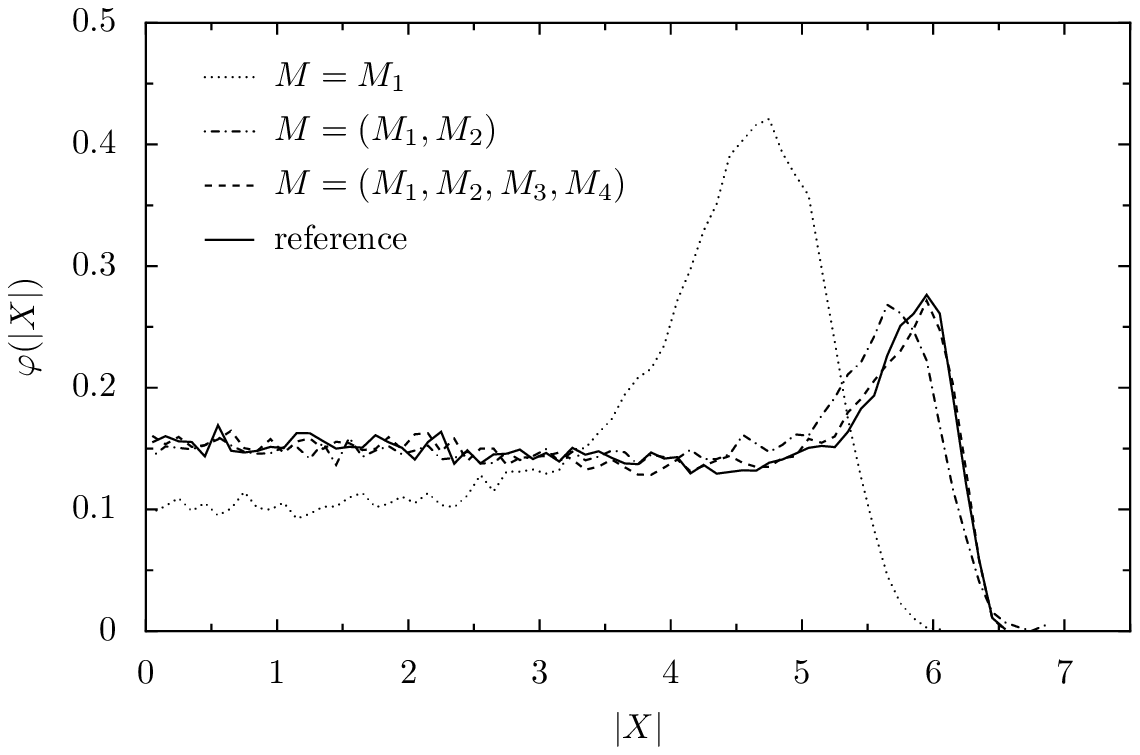}
	\includegraphics[scale=0.6]{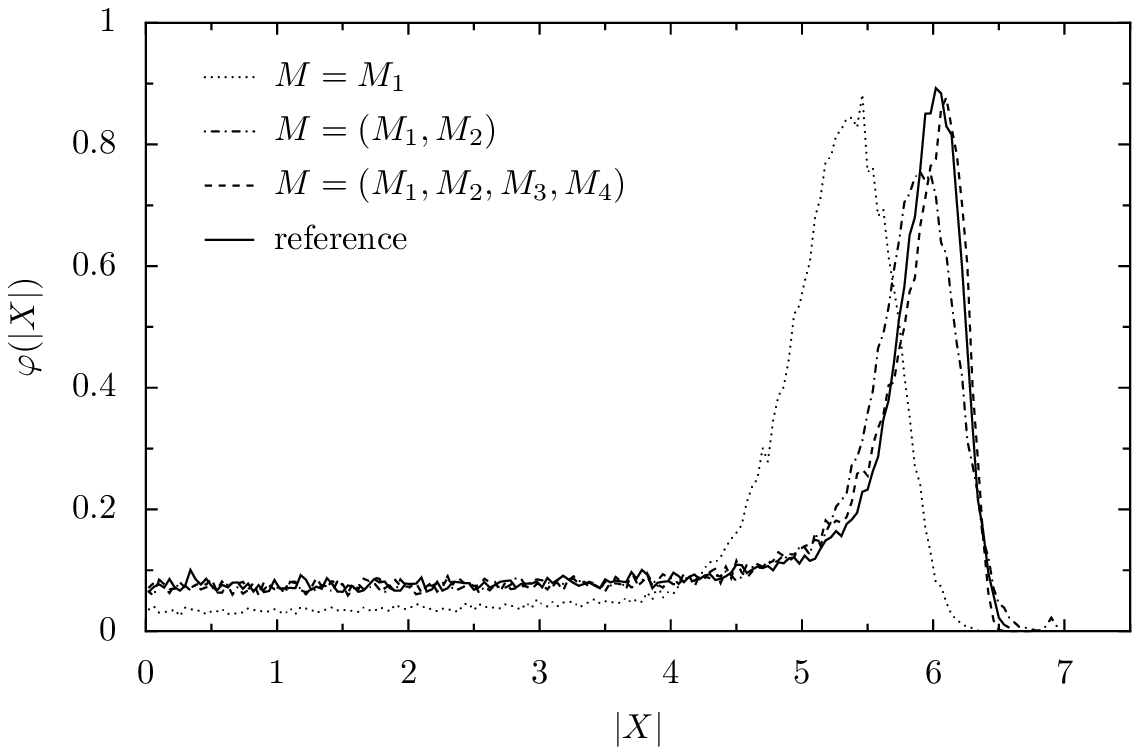}
	\includegraphics[scale=0.6]{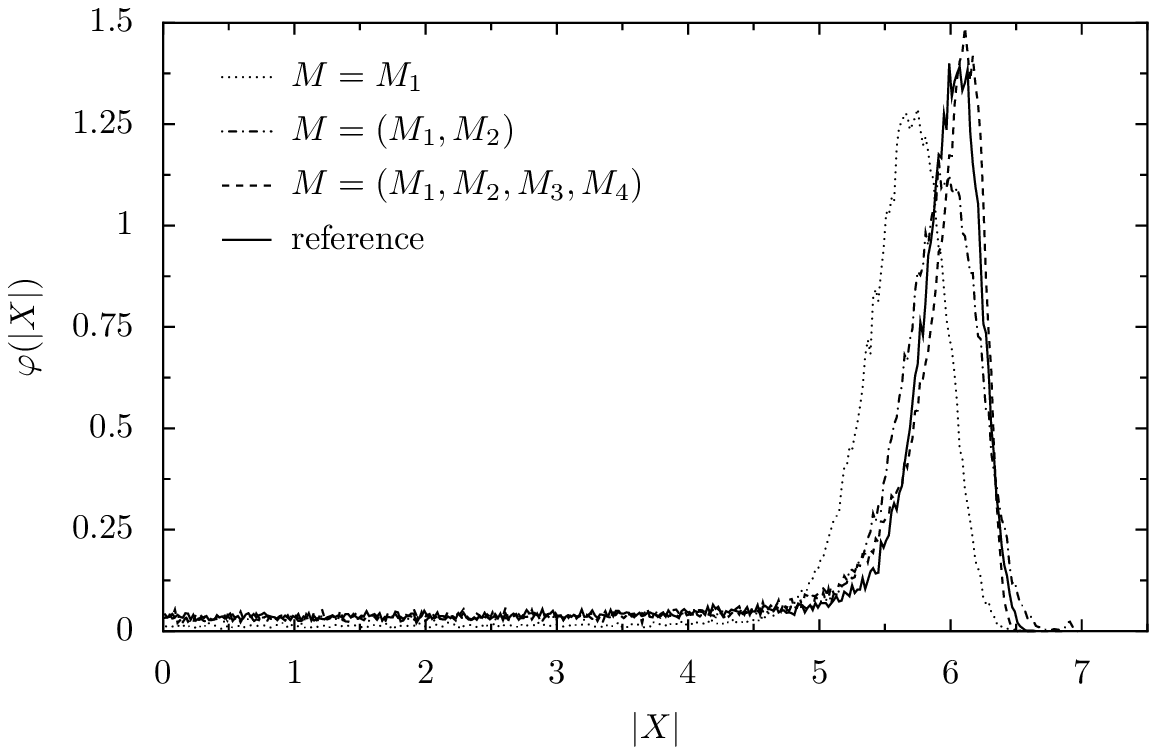}
\end{center}
\caption{\label{fig:stressmu_constrained} Lifted polymer distributions as a function of the number of macroscopic state variables using strategy 3.  We plot a reference polymer distribution, {\em i.e.}, the polymer distribution after a microscopic simulation up to time $t^*$, as well as the equilibrium polymer distributions after constrained simulation using the indicated macroscopic state variables. Shown are the results for $t^*=0.5$ (top left), $t^*=1$ (top right), $t^*=1.5$ (bottom left) and $t^*=2$ (bottom right). Simulation parameters are given in the text.}
\end{figure}
Compared to the previous two strategies, we here observe very good agreement with the reference distribution with less macroscopic state variables.

\subsubsection{Coarse time-stepping}

We now look at the evolution of the numerical closure with respect to the full microscopic simulation, again using $\kappa(t)=2$.  We simulate an ensemble of $N=2000$ FENE dumbbells, starting from the equilibrium distribution in the absence of flow up to time $t=4$ and compare this reference simulation with a number of simulations using the coarse time-stepper with a different number $L$ macroscopic state variables as above. As before, we choose the macroscopic time-step equal to one microscopic step $\delta t$, {\em i.e.}, $K=1$; all other parameters are also chosen as above.  We lift by freezing physical time and performing a constrained simulation that is consistent with $\M$ until equilibrium of the distribution is reached (here using $m_\infty=50$ constrained time-steps of size $\delta t$); all simulations were verified to have converged with respect to the number of constrained time-steps.  The results are shown in Figure~\ref{fig:cts_stressmu_startup}.
\begin{figure}
\begin{center}	\includegraphics[scale=0.75]{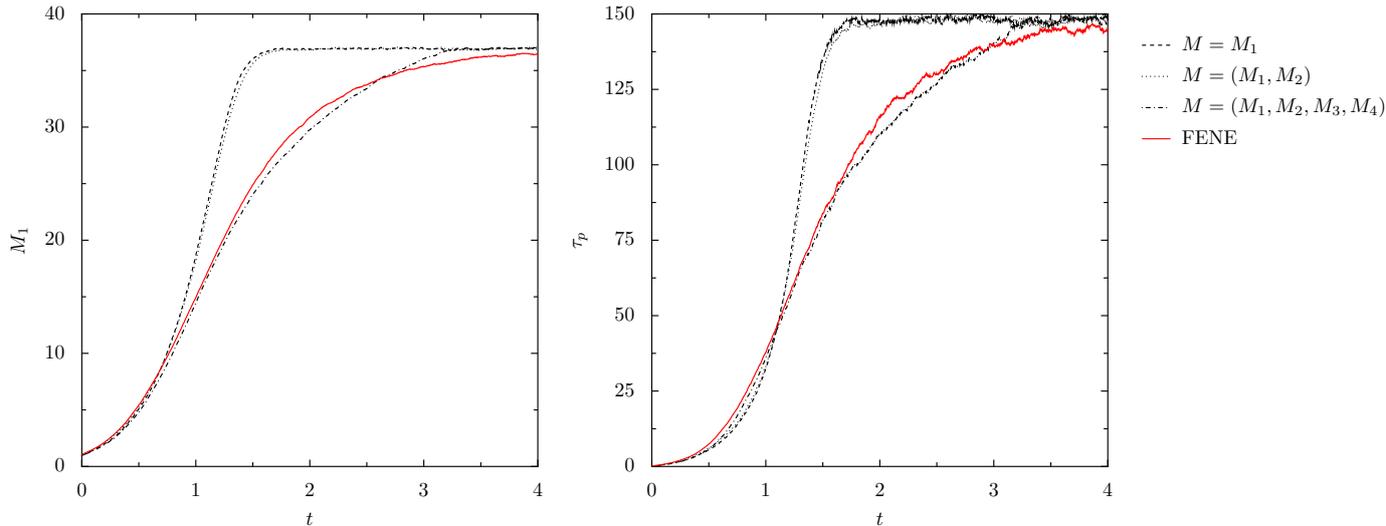}
\end{center}
\caption{\label{fig:cts_stressmu_startup}  Evolution of the first even moment $M_1$ (left) and stress $\tau_p$ (right) for an ensemble of FENE dumbbells during startup of elongational flow. Shown are a full microscopic simulation (reference), and simulations using a coarse time-stepper for different numbers of macroscopic state variables using strategy 3.  Simulation parameters are given in the text.}
\end{figure}
Also here, we see the improvement; the result of the complex flow experiment is shown in Figure~\ref{fig:fene_stressmu_cts}.
\begin{figure}
\begin{center}
		\includegraphics[scale=0.75]{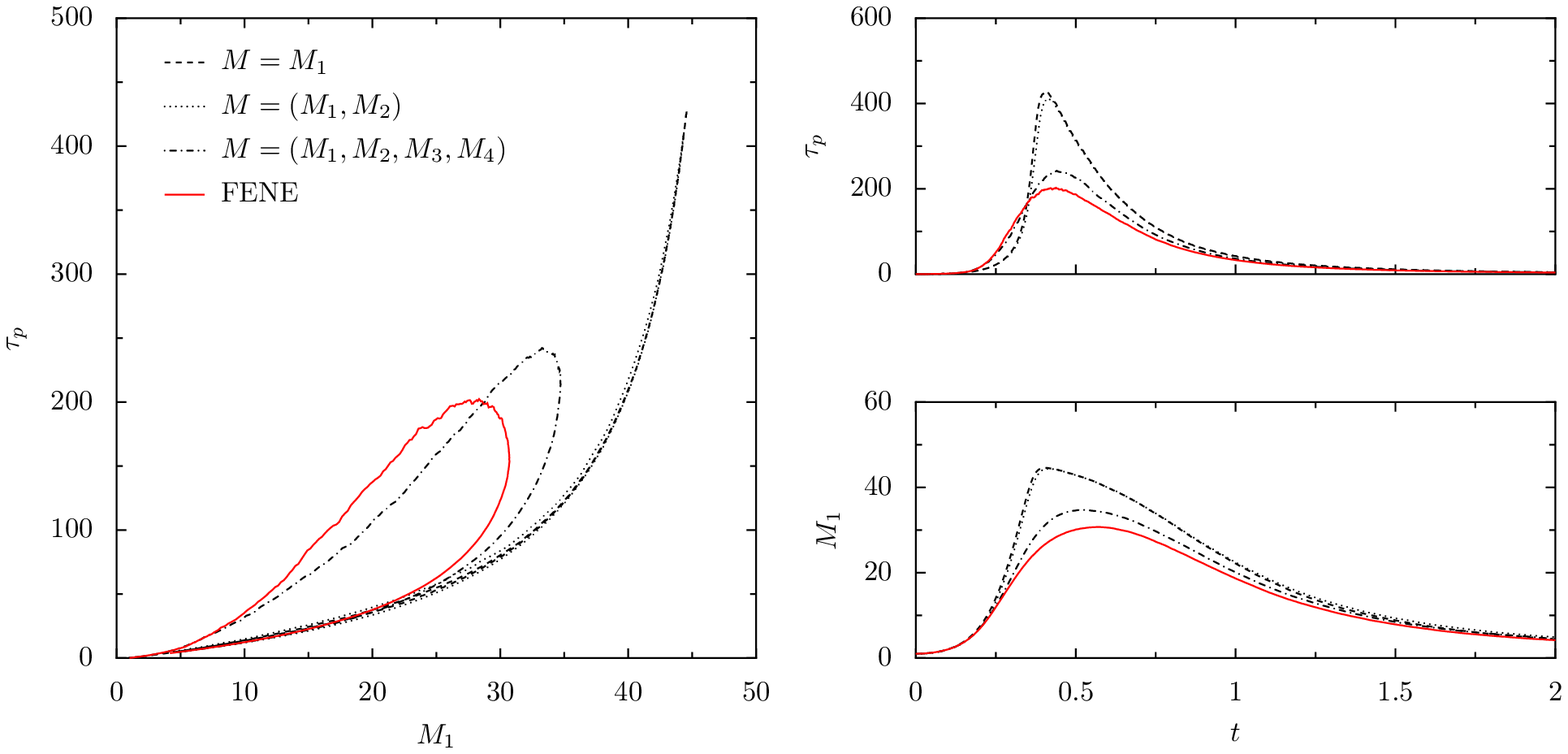}
\end{center}
\caption{\label{fig:fene_stressmu_cts}  Evolution of first even moment $M_1$ and stress $\tau_p$ for an ensemble of FENE dumbbells during complex flow. Left: $(M_1,\tau_p)$ phase plane view.  Right: temporal evolution. Shown are a full microscopic simulation (reference), and simulations using a coarse time-stepper for different numbers of macroscopic state variables using strategy 3.  Simulation parameters are given in the text.}
\end{figure}

\section{Conclusions and discussion\label{sec:concl}}

We proposed a numerical closure strategy that enables to easily explore which sets of macroscopic state variables should be chosen to get good closure approximations for the kinetic simulation of polymeric fluids. The method involves the reconstruction of a polymer distribution as the constrained equilibrium of a microscopic Monte Carlo simulation, constrained upon the desired macroscopic state. The resulting algorithm is very flexible, and enables to explore the error introduced by the closure for various sets of macroscopic state variables $\M$. 
We showed that this numerical closure approximation is optimal, in the sense that, when applied to a microscopic model which has an equivalent macroscopic model, it indeed yields the macroscopic model. Moreover, in some specific cases, the approach is shown to be closely related to the closure approximation based on a quasi-equilibrium condition. 
While the exposition in the present paper was restricted to the one-dimensional case, extensions to higher space dimensions are straightforward.

The procedure straightforwardly enables to test hypotheses on which macroscopic state variables should be included to build good closures. We have examined three strategies to define a hierarchy of macroscopic state variables. Our numerical experiments indicate that, at least for the cases considered in this paper, fewer macroscopic state variables are required to obtain accurate results when choosing a strategy that adds macroscopic state variables based on the unknowns that appear on the right-hand side of an It\^o calculation for the already included state variables (Strategy 3 in this text). Moreover, the experiments in section \ref{sec:relax} indicate that, in principle, the accuracy of the numerical closure can be estimated by monitoring non-constrained state variables during the constrained simulation.  Finally, when one can accurately assess the (lack of) accuracy of a given set of macroscopic state variables, it is straightforward to adjust the number of macroscopic state variables throughout a simulation using a corresponding accuracy criterion, as is done in \cite{TonyAdaptive,IlgAdapt}. Note that, once a good set of macroscopic state variables is obtained, one could also consider proceeding along the lines of \cite{Ilg:2002p10825} to obtain a quasi-equilibrium closure. 

So far, we have not discussed potential	gains in computational efficiency compared to a full microscopic simulation.  One way to achieve a reduction in computational cost is to make use of coarse projective integration \cite{KevrGearHymKevrRunTheo03,Kevrekidis:2009p7484} or similar methods \cite{EEng03,E:2007p3747}. In this type of methods, one uses the proposed numerical closure technique to estimate the \emph{time derivative} of the unavailable macroscopic model, and uses this estimated time derivative to perform a large (projective) forward Euler step for the macroscopic state variables; one then repeats the numerical closure procedure. The efficiency of coarse projective integration is strongly tied to a separation in time-scales between relaxation and macroscopic evolution; unfortunately, the physically interesting non-Newtonian behaviour precisely appears when this time-scale separation is absent.  We refer to  \cite{LiVdE} for a study of the acceleration that can be obtained in the small Deborah number limit, in which the polymeric fluid becomes Newtonian.

\section*{Acknowledgments}
GS is a Postdoctoral Fellow of the Research Foundation – Flanders (FWO – Vlaanderen). This work was performed during a research stay of GS at the Department of Mechanical Engineering, iMMC, U.C. Louvain, whose hospitality is gratefully acknowledged. TL would like to thank S. Olla for fruitful discussions. GS would like to thank R. Keunings for stimulating discussions. The work of GS was supported by the Research Foundation – Flanders through Research Project G.0130.03 and by the Interuniversity Attraction Poles Programme of the Belgian Science Policy Office through grant IUAP/V/22. The scientific responsibility rests with its authors.

\bibliographystyle{plain}
\bibliography{polymers.bib}
\end{document}